\documentclass[nonacm]{acmart}

\AtBeginDocument{%
  \providecommand\BibTeX{{%
    \normalfont B\kern-0.5em{\scshape i\kern-0.25em b}\kern-0.8em\TeX}}}

\setcopyright{acmcopyright}
\copyrightyear{2021}
\acmYear{2021}
\acmDOI{}

\acmConference{}
\acmBooktitle{}
\acmPrice{}
\acmISBN{}
\usepackage{enumitem}
\usepackage{subcaption}

\begin{document}

\title{Harmless but Useful: Beyond Separable Equality Constraints in Datalog+/-}

 \author{Luigi Bellomarini}
 \affiliation{
   \institution{Banca d'Italia}
 }

 \author{Emanuel Sallinger}
 \affiliation{
   \institution{TU Wien \& University of Oxford}
}

\renewcommand{\shortauthors}{L. Bellomarini, E. Sallinger}

 \begin{abstract}
 Ontological query answering is the problem of answering queries in the presence of schema constraints representing the domain of interest.
 Datalog+/- is a commonly adopted family of languages for schema constraints, which includes tuple-generating dependencies (TGDs) and equality-generating dependencies (EGDs). Unfortunately, the interplay of TGDs and EGDs leads to undecidability or intractability of query answering when adding EGDs to tractable Datalog+/- fragments, like Warded Datalog+/-, for which, in the sole presence of TGDs, query answering is PTIME in data complexity. There have been attempts to limit the interaction of TGDs and EGDs and guarantee tractability, in particular with the introduction of separable EGDs, whose idea is making EGDs irrelevant for query answering as long as the set of TGD constraints is satisfied. While being tractable, separable EGDs have limited expressive power.
 
 In this paper, we propose a more general class of EGDs, which we call ``harmless'', that subsume separable EGDs and allow to model and reason about a much broader class of problems. Unlike separable EGDs, harmless EGDs do affect query answering in the sense that, besides enforcing ground equality constraints, they specialize the query answer by grounding or renaming the labelled nulls introduced by existential quantification in the TGDs. Harmless EGDs capture the cases when the answer obtained in the presence of EGDs is strictly less general than ---or an image of--- the one, if any, obtained with TGDs only.
 We study the theoretical problem of deciding whether a set of constraints contains harmless EGDs and conclude it is undecidable. Nevertheless, we contribute a sufficient syntactic condition characterizing harmless EGDs, which is broad and useful in practice. 
 We focus on the Warded Datalog+/- fragment and argue that, in such language, query answering keeps decidable and PTIME in data complexity in the presence of harmless EGDs. 
 We study principled chase-based techniques for query answering in Warded Datalog+/- with harmless EGDs, conducive to an efficient algorithm to be implemented in state-of-the-art reasoners.
 \end{abstract}

\maketitle

\section{Introduction}
\label{sec:introduction}

Given a query $Q$, a database $D$ and a set of schema constraints $\Sigma$,  \textit{ontological query answering} refers to the problem of computing all the possible sets of facts $B$ that satisfy $Q$ such that $B \supseteq D$ and $B$ satisfies $\Sigma$. In other words, not only must the query be computed with respect to $D$, but it should also include all the facts that are entailed from $D$ via $\Sigma$.

Datalog$^\pm$~\cite{CaGL09} is a general class of languages for schema constraints, which extends Datalog~\cite{datalog2}: a set of Datalog$^\pm$ rules is a set of \textit{function-free Horn clauses},  potentially including \textit{existential quantification}, i.e., tuple-generating dependencies (TGDs), negative constraints with the falsum ($\bot$) in rule head, stratified negation, and \textit{functionality constraints} expressed as \textit{equality-generating dependencies} (EGDs). 

\medskip
A TGD is a first-order implication of the form $\forall \mathbf{x}~\boldsymbol{\phi}(\mathbf{x}) \to \exists\mathbf{z}~\boldsymbol{\psi}(\mathbf{y},\mathbf{z})$, where $\boldsymbol{\phi}(\mathbf{x})$ and $\boldsymbol{\psi}(\mathbf{y},\mathbf{z})$ are conjunctions of atoms over a relational schema. An EGD is a first-order implication of the form 
$\forall \mathbf{x}~\boldsymbol{\phi}(\mathbf{x}) \to x_i=x_j$, where $\boldsymbol{\phi}(\mathbf{x})$ is a conjunction of atoms over a relational schema and $x_i,x_j \in \mathbf{x}$.
When performing ontological query answering, via the so-called \textsc{chase} procedure~\cite{FKMP05}, we expand $D$ with facts entailed via the application of rules $\Sigma$ and introduce placeholder variables known as \textit{labelled nulls} or \textit{marked nulls} to denote the unknown values generated by existential quantification. Finally, we evaluate $Q$ upon these facts. 
EGDs affect query answering by enforcing equalities, which can result in either comparisons between constants, satisfied or causing the chase failure, or renaming of labelled nulls.

If $\Sigma$ is a set of TGDs, in the presence of recursion and existential quantification, ontological query answering is in general undecidable~\cite{GoPi15}. Thus, Datalog$^\pm$ fragments adopt syntactic restrictions to guarantee decidability and tractability, like in the fragment denoted by the eponymous \textit{Warded Datalog$^\pm$}~\cite{GoPi15} language, where only a limited propagation of labelled nulls is allowed so as to keep query answering decidable and \textsf{PTIME} in data complexity. Unfortunately, the interplay of TGDs and EGDs, when there is no limitation in how the two can interact, causes undecidability of query answering even in elementary cases~\cite{ChVa85}. This paper proposes a language for EGDs that allows high interaction with TGDs while guaranteeing decidability and tractability in \textit{Warded Datalog$^\pm$}, thus a broad fragment of Datalog$^\pm$.

\medskip
\noindent Let us consider the following set of rules from an assembly line domain (where we omit universal quantifiers to simplify):
\begin{example}
\label{ex:intro_example}
\begin{align*}
\textit{component}(x) \to \exists z~\textit{component}(z), \textit{partOf}(x,z) \tag{$\sigma_1$} \\
\textit{partOf}(x,v), \textit{partOf}(x,w) \to v=w \tag{$\eta_1$} \\
\textit{component}(x),\textit{component}(y), \textit{tag}(x,y), \textit{partOf}(x,v), \textit{partOf}(y,w) \to v=w \tag{$\eta_2$}
\end{align*}

\textit{Components are parts of other components, as denoted by $\sigma_1$.
EGD $\eta_1$ asserts that every component can be part of one single component and, by $\eta_2$, two components sharing a tag are part of the same component.~$\blacksquare$}
\end{example}
Now, let us take 
$D=\{\mathit{component}(\mathit{engine}),
\mathit{component}(\mathit{piston}), \mathit{component}(\mathit{camshaft}), \mathit{component}(\mathit{lobe}),$\\
$\mathit{component}(\mathit{thrust}), 
\mathit{partOf}(\textit{piston},\textit{engine}),
\mathit{partOf}(\textit{lobe},\textit{camshaft}),
\mathit{tag}(\textit{piston},\textit{camshaft}),
\mathit{tag}(\textit{lobe},\textit{thrust})\}$
and the gro\-und Boolean conjunctive query $Q$ defined as 
$q \leftarrow \mathit{partOf}(\mathit{thrust},\mathit{camshaft}),\mathit{partOf}(\mathit{camshaft},\mathit{engine})$ 
where $q$ is a propositional atom satisfied if the thrust is part of the engine via the camshaft. Now, we want to answer $Q$ by expanding $D$ with $\Sigma$. By $\sigma_1$ we derive 
$\mathit{partOf}(\textit{camshaft},\nu_1)$, $\mathit{partOf}(\textit{thrust},\nu_2)$,$\mathit{partOf}(\nu_2,\nu_3)$, along with the corresponding facts $\textit{component}(\nu_1)$, $\textit{component}(\nu_2)$, $\textit{component}(\nu_3),$ where $\nu_1, \nu_2, \nu_3$ are labelled nulls generated by existentials. By $\eta_2$ we can conclude $v_1 = \mathit{engine}$ as the camshaft and the piston must belong to the same component (the engine), and $\nu_2=\textit{camshaft}$, since also the lobe and the thrust are bound to be part of the same component (the camshaft). This is already enough for a positive answer to $Q$ and, by $\eta_1$, we can also confirm that $\nu_3= \mathit{engine}$.

\medskip
Observe that the solution to $Q$ obtained by only applying TGD $\sigma_1$ can be mapped onto the solution obtained also considering $\eta_1$ and $\eta_2$ by a substitution of variables $h=\{\nu_1\to\mathit{engine}, \nu_2\to\textit{camshaft}, \nu_3\to \mathit{engine}\}$. For a set of ``\textit{harmless}'' EGDs respecting specific syntactic conditions, like $\eta_1$ and $\eta_2$, this property holds for every $Q$ that admits an answer. Also note that in the absence of the EGDs, $\eta_2$ in particular, a positive answer to $Q$ would not be possible. This witnesses that harmless EGDs allow for an expressive form of interaction with TGDs, more general than the existing fragment of \textit{separable} EGDs~\cite{CaCF13}, which encode a lack of interaction, hence a separation, between TGDs and EGDs. 

\medskip\noindent\textbf{Contribution}. This paper presents the following contribution.

\begin{enumerate}[leftmargin=2mm]\setlength\itemsep{1em}
\item We introduce \textbf{the class of harmless EGDs}, which subsume separable EGDs and allow to model and reason upon a broader class of problems. In particular, we show that every separable EGD is also harmless (and not vice versa).
We prove that the problem of establishing whether a set of EGDs in $\Sigma$ is harmless is undecidable. We contribute a \textbf{sufficient syntactic condition}, namely \textit{safe taintedness}, which witnesses harmless EGDs in many practical cases.
It goes without saying that EGDs have extreme relevance in the data exchange~\cite{FKPT05} and reasoning literature~\cite{CaGL12}. Interestingly,
they are used in recent extensive benchmark results for state-of-the-art reasoners such as Vadalog~\cite{BeSG18}, as well as industrial and synthetic scenarios including ChaseBench~\cite{BKMM17}, a benchmark targeting data exchange and query answering problems, and iBench~\cite{AGCM15}, a metadata generator for data exchange settings. 
EGDs also play a key role in a large set of industrial scenarios~\cite{BFGK18}. Let us emphasise that we could verify that more than 95\% of EGDs adopted in the mentioned settings do respect our definition of harmless and are also captured by our syntactic condition.
\item We deal with \textbf{query answering with harmless EGDs} in Warded Datalog$^\pm$ and prove the problem is decidable and \textsf{PTIME} in data complexity.
Beyond the purely theoretical results, we face the challenge of providing practical algorithms for conjunctive query answering with harmless EGDs. To this aim, we move from the observations made for Warded Datalog$^\pm$~\cite{BeSG18} (\textit{warded semantics}), which allows to define a form of chase where isomorphic portions of the chase can be considered only once, keeping the chase finite and of limited size, and propose the \textit{relaxed warded semantics}, which underpins a \textbf{new variant of the chase for harmless EGDs} that extends the techniques used for Warded Datalog$^\pm$. Our chase exploits the specific form of interaction between TGDs and harmless EGDs to perform query answering. 
\item Towards a prototype implementation of harmless EGDs in the \textsc{Vadalog} system, a state-of-the-art Warded Datalog$^\pm$ reasoner~\cite{BeSG18}, we provide a \textbf{practical algorithm} based on the theoretical underpinnings of harmless EGDs, which implements our new variant of the chase.
\end{enumerate}

\medskip\noindent\textbf{Overview}. The remainder of the paper is organized as follows. In Section~\ref{sec:preliminaries} we recall the preliminary notions. In Section~\ref{sec:harmless} we introduce harmless EGDs, showing their relationship with the separable ones and the related work. In Section~\ref{sec:harmless_harmful} we develop the theory behind harmless EGDs, showing the decidability results and the syntactic conditions. Section~\ref{sec:efficient_evaluation} deals with query answering with Datalog$^\pm$ and harmless EGDs and lays the basis for the implementations by providing a practical algorithm. In Section~\ref{sec:conclusion} we draw up our conclusions.

\section{Preliminaries}
\label{sec:preliminaries}
Let us start by laying out the preliminary notions.

\medskip\noindent\textbf{Relational foundations}.
Let $\mathbf{C}$, $\mathbf{N}$, and $\mathbf{V}$ be disjoint countably infinite sets of {\em constants}, {\em (labelled) nulls} and {\em variables}, respectively.  A {\em term} is a either a constant or variable. Different constants represent different values (\textit{unique name assumption}). A {\em (relational) schema} $\mathbf{S}$ is a finite set of relation symbols with associated arity. Given a schema $\mathbf{S}$, an {\em atom} is an expression of the form $R(\bar v)$, where $R \in \mathbf{S}$ is of arity $n > 0$ and $\bar v$ is an $n$-tuple of terms. A {\em database instance} (or simply \textit{database}) over $\mathbf{S}$ associates to each relation symbol in $\mathbf{S}$ a relation of the respective arity over the domain of constants and nulls.
The members of relations are called \textit{tuples} or \textit{facts}. Sometimes, with some abuse of notation, we will use the terms tuple, fact and atom interchangeably. 
A \textit{substitution} from one set of symbols $S_1$ to another set of symbols $S_2$ is a function $h : S_1 \to S_2$, defined as a possibly empty set of assignments $\{X\to Y\}$ (and we also write $Y=h(X)$ and $S_2=h(S_1)$ with intuitive meaning), where $X \in S_1$ and $Y \in S_2$ and each symbol of $S_1$ can be mapped at most to one symbol of $S_2$. 
Given two sets of atoms $A_1$ and $A_2$, we define a \textit{homomorphism} from $A_1$ to $A_2$, a substitution 
$h : \mathbf{C}\cup\mathbf{N}\cup\mathbf{V} \to \mathbf{C}\cup{N}\cup\mathbf{V}$ such that $h(t)=t$, if $t \in \mathbf{C}$, and for each atom $a(t_1,\ldots,t_n)\in A_1$, we have that $h(a(t_1,\ldots,t_n))=a(h(t_1),\ldots,h(t_n))$ is in $A_2$.  
We define a homomorphism from $A_1$ to $A_2$ as \textit{partial} if it maps a subset $B$ of $A_1$ (possibly $A_1$ itself) to $A_2$, i.e., $h(B)$ coincides with $A_2$, with $B\subseteq A_1$. 
We say that a homomorphism $h$ maps $A_1$ onto $A_2$,
whenever $h(A_1)$ coincides with $A_2$ (i.e., $h$ is surjective). In the rest of the paper we will in general refer to homomorphisms with this property unless otherwise specified. The definition of homomorphism 
easily extends to conjunctions of atoms. 
We will also adopt the function composition formalism $h^\prime\circ h$ to denote the orderly application of the homomorphisms or substitutions $h$ and $h^\prime$ (i.e. $h^\prime(h(\ldots))$. Two sets of atoms $A_1$ and $A_2$ are \textit{isomorphic} if there exists a homomorphism $h$ such that $A_1$ can be mapped onto $A_2$ and vice versa.

\medskip\noindent\textbf{Conjunctive queries}.
A \textit{conjunctive query} (CQ) $Q$ over a schema $\mathbf{S}$ is an implication of the form $q(\mathbf{x}) \leftarrow \boldsymbol{\phi}(\mathbf{x},\mathbf{y})$, where $\boldsymbol{\phi}(\mathbf{x},\mathbf{y})$ is a conjunction of atoms over $\mathbf{S}$, $q(\mathbf{x})$ is an atom that does not occur in $\mathbf{S}$, and $\mathbf{x}$ and $\mathbf{y}$ are vectors of terms. If $q$ is of arity zero, then $Q$ is a \textit{Boolean conjunctive query} (BCQ).
An answer to a conjunctive query $Q$ : $q(\mathbf{x}) \leftarrow \boldsymbol{\phi}(\mathbf{x},\mathbf{y})$ of arity $n$, over a database $D$, is the set $Q(D)$ of all the facts $t$, for which there exists a homomorphism $h: \mathbf{x}\cup\mathbf{y} \to \mathbf{C}\cup\mathbf{N}$ such that $h(\boldsymbol{\phi}(\mathbf{x},\mathbf{y})) \subseteq D$ and $t = h(\mathbf{x})$. In BCQs, q is a propositional atom, and we write $D\models Q$ if $Q(D)$ is not empty.

\medskip\noindent\textbf{Dependencies}. 
Datalog\(^\pm\) is a family of languages for schema constraints that extends Datalog with existential quantification and other minor features such as negative constraints with falsum ($\bot$) in the head and stratified negation. A set of Datalog$^\pm$ rules is a set of tuple-generating dependencies (TGDs). A TGD is a first-order implication of the form $\forall \mathbf{x}~\boldsymbol{\phi}(\mathbf{x}) \to \exists\mathbf{z}~\boldsymbol{\psi}(\mathbf{y},\mathbf{z})$, where $\boldsymbol{\phi}(\mathbf{x})$ (the \textit{body}) and $\boldsymbol{\psi}(\mathbf{y},\mathbf{z})$ (the \textit{head}) are conjunctions of atoms over a relational schema and boldface variables denote vectors of variables, with $\mathbf{x}\subseteq\mathbf{y}$. For brevity, we write these existential rules as $\boldsymbol{\phi}(\mathbf{x}) \to \exists~\mathbf{z}~\boldsymbol{\psi}(\mathbf{y},\mathbf{z})$, using commas to denote conjunction of atoms in $\boldsymbol{\phi}(\mathbf{x})$. A TGD $\sigma$ is satisfied by a database $D$ of schema $\mathbf{S}$ (and we write $D \models \sigma$) if whenever there is a homomorphism $h$ such that $h(\boldsymbol{\phi}(\mathbf{x})) \subseteq D$, there exists an \textit{extension} $h^\prime$ of $h$ (i.e., $h \subseteq h^\prime$) such that $h(\boldsymbol{\psi}(\mathbf{y},\mathbf{z})) \subseteq D$.
An \textit{equality-generating dependency} (EGD) is a first-order implication of the form $\forall \mathbf{x}~\boldsymbol{\phi}(\mathbf{x}) \to x_i=x_j$, where $\boldsymbol{\phi}(\mathbf{x})$ is a conjunction of atoms over a relational schema and $x_i,x_j \in \mathbf{x}$. We will omit universal quantifiers also in this case. A database $D$ of schema $\mathbf{S}$ satisfies an EGD $\eta$ if whenever there is a homomorphism $h$ such that $h(\boldsymbol{\phi}(\mathbf{x})) \subseteq D$, then we have that $h(x_i)=h(x_j)$.

\medskip\noindent\textbf{Ontological query answering}. The definition of query answering is usually extended to account for schema constraints. Given a database $D$ of schema $\mathbf{S}$ and a set $\Sigma$ of TGDs and EGDs over $\mathbf{S}$, we name the \textit{models} of $D$ and $\Sigma$ as the set of all databases $B$ (and we write $B \models D \cup \Sigma$) such that $B \supseteq D$, and $B \models \Sigma$. The answer to a CQ $Q$ over $D$ under $\Sigma$ is the set of facts $t$ such that $t \in Q(B)$, where $B \models D \cup \Sigma$. Also in this case, a positive answer to a BCQ ($D\cup\Sigma \models Q$) corresponds to a non-empty set of terms $t \in Q(B)$.
Query answering under general TGDs is undecidable even when $Q$ is fixed and $D$ is given as input~\cite{CaGK13}. So, different restrictions to the language of TGDs have been introduced and are reflected in the fragments of Datalog$^\pm$. Warded Datalog$^\pm$ is one of them, with a very good trade-off between expressive power and computational complexity, with CQ in fact being \textsf{PTIME}~\cite{BeSG18}. In the presence of EGDs, query answering under TGDs is also undecidable, even in very simple cases where EGDs are used to define key constraints~\cite{ChVa85}. Datalog$^\pm$ fragments allow only very limited forms of interaction between TGDs and EGDs, such as \textit{separable EGDs}~\cite{CaCF13}, which do not hamper tractability and decidability, as we discuss in detail in Section~\ref{sec:harmless}. In the rest of the paper we will refer to BCQs without loss of generality. In fact, the problems of CQ answering and BCQ answering under TGDs and EGDs are \textsf{LOGSPACE}-equivalent as the decision version of CQ answering and BCQ answering are mutually \textsf{AC$_0$}-reducible~\cite{CaGK13}.

\medskip\noindent\textbf{The Chase}. Chase-based procedures~\cite{MaMS79} modify a database $D$ by adding facts to it, until it satisfies a set of constraints $\Sigma$. Intuitively, in the context of ontological query answering, given a BCQ Q, the chase expands $D$ with facts inferred by applying $\Sigma$ to $D$ into a database $\textit{chase}(\Sigma,D)$ (which we will refer to as ``the chase'' as well, abusing terminology), possibly containing labelled nulls and such that $\textit{chase}(\Sigma,D) \models Q$. In particular a chase execution $\textit{chase}(\Sigma,D)$ builds a \textit{universal model} for $D$ and $\Sigma$, i.e., for every database $B$ that is a model for $D$ and $\Sigma$, there is a homomorphism mapping $\textit{chase}(\Sigma,D)$ onto $B$.
The chase fails if constraints in $\Sigma$ cannot be satisfied by expanding $D$. Many forms of the chase exists, and we will focus on specific variants when needed in the discussion. Let us recall for the moment two general working rules, the \textit{TGD chase step} and the \textit{EGD chase step}. Given a database $D$ over a schema $\mathbf{S}$, and a TGD $\sigma: \boldsymbol{\phi}(\mathbf{x}) \to \exists\mathbf{z}~\boldsymbol{\psi}(\mathbf{y},\mathbf{z})$, we have that $\sigma$ is applicable if there exists a homomorphism $h$ such that $h(\boldsymbol{\phi}(\mathbf{x})) \subseteq D$. Then, the TGD chase step, adds the fact $h^\prime(\boldsymbol{\psi}(\mathbf{y},\mathbf{z}))$ to $D$, if not already in $D$, where $h^\prime \supseteq h$ is a homomorphism that extends $h$ and maps variables of $\mathbf{z}$ to newly created labelled nulls. In the same setting, given an EGD $\eta: \boldsymbol{\phi}(\mathbf{x}) \to x_i=x_j$, we have that $\eta$ is applicable if there exists a homomorphism $h$ such that $h(\boldsymbol{\phi}(\mathbf{x})) \subseteq D$ and $h(x_i)\neq h(x_j)$. Then, the EGD chase steps proceeds as follows: (i)~if both $x_i$ and $x_j$ are constants, it fails; (ii)~if $x_i$ (resp.\ $x_j$) is a variable, the chase steps replaces each occurrence of $x_i$ (resp.\ $x_j$) with $x_j$ (resp.\ $x_i$) in $D$.

With the definition of chase steps in place, we are ready to recall the chase procedure. Given a database $D$ and a set of constraints $\Sigma = \Sigma_T \cup \Sigma_E$ (TGDs and EGDs, respectively), the chase applies each applicable TGD steps (once) and all the EGD steps to fixpoint, i.e., as long as applicable or a hard constraint is violated. 
We can see the chase in action by considering again Example~\ref{ex:intro_example}. With a TGD chase step we generate $\textit{component}(\nu_1)$ and $\textit{partOf}(\textit{camshaft},\nu_1)$, where $\nu_1$ is a labelled null. We try to apply all EGDs to fixpoint and, in particular, $\eta_2$ assigns $\nu_1 = \textit{engine}$ because of the homomorphism $h = \{x\to \textit{piston}, y\to\textit{camshaft}, v\to\textit{engine},w\to\nu_1\}$ from the body of $\eta_2$ to $D$, which implies the assignment of $\nu_1$ to the value of $v$. We then proceed with another TGD step and obtain $\textit{component}(\nu_2)$ and $\textit{partOf}(\textit{thrust},\nu_2$); by applying again $\eta_2$ with EGD steps, we can assign $\nu_2 = \textit{camshaft}$, and so on.

\section{Adding EGDs to Reasoning in Datalog$^\pm$}
\label{sec:harmless}

Towards our proposal for a new class of more expressive EGDs, we first discuss the main limitation of the related approaches, separable EGDs in particular.

\subsection{Prior Work}
As the interaction between EGDs and TGDs leads to undecidability in query answering~\cite{ChVa85}, in the Datalog$^\pm$ context, there have been some attempts to introduce semantic restrictions to EGDs~\cite{CaGK08,CaGP10,CaPi11,CaGL12} to make them tractable. 
An initial intuition is that of \textit{innocuous EGDs}~\cite{CaGK08}. They enjoy the property that query answering is insensitive to them, provided that the chase does not fail. As a consequence, given a relational schema $\mathbf{S}$ and a set $\Sigma = \Sigma_T \cup \Sigma_E$ of TGDs and EGDs over $\mathbf{S}$, a chase application can simply ignore $\Sigma_E$, with the guarantee that all the facts needed for query answering will be entailed. This property, which is semantic and cannot be syntactically checked, is of scarce practical utility, as the adopted EGDs are superfluous and do not add expressive power to the TGDs.

A more interesting notion is that of \textit{separable EGDs}~\cite{CaCF13}. A set of constraints $\Sigma$ is separable if for every database $D$ for $\mathbf{S}$: (i) if the chase of $\Sigma$ over $D$ fails, then $D$ does not satisfy $\Sigma_E$; (ii) if the chase does not fail, then we have that $\text{chase}(\Sigma,D) \models Q$ iff $\text{chase}(\Sigma_T,D) \models Q$ for every BCQ $Q$ over $\mathbf{S}$. 
Essentially, separable EGDs do not interact with the TGDs. This concept was originally introduced in the context of inclusion dependencies and key dependencies~\cite{AbHV95}, and has been reformulated more recently also under the notion of \textit{EGD-stability}~\cite{CaCF13} according to which a set of TGDs and EGDs $\Sigma$ is EGD-stable if, for every instance $D$ of $\mathbf{S}$, if $D$ satisfies $\Sigma_E$, then the chase of $D$ under $\Sigma$ does not fail. 

While it has been proven that in general checking whether a set of rules $\Sigma$ is separable is undecidable, specific cases of separability can be syntactically checked. In particular, \textit{non-conflicting} sets of TGDs and EGDs, a definition given in the context of functional dependencies, have been shown to be separable~\cite{CaGP10}.

\subsection{Beyond Separability}
The main deficiency of separable EGDs is their \textit{limited unification power}. So, while syntactically, separable EGDs appear to equate also variables of intensional predicates, in fact they never cause distinct labelled nulls to be unified ---motivating why separable EGDs can be verified a priori against $D$--- and do not contribute to altering the TGD chase. To better grasp this limitation, let us consider the following two examples:

\begin{example}
\label{ex:first}
\begin{align*}
  \textit{element}(x) \rightarrow \exists z~\textit{comp}(x,z) \tag{$\sigma_1$} \\
  \textit{att}(x,k),\textit{att}(y,k),\textit{comp}(x,z^\prime),  \textit{comp}(y,z^{\prime\prime}) \to z^\prime=z^{\prime\prime} \tag{$\eta_1$}\\
\end{align*}
\textit{A basic clustering scenario: every element $x$ is a component of a set $z$ ($\sigma_1$) and whenever two elements $x$ and $y$ share an attribute $k$, then they belong to the same set ($\eta_1$).~$\blacksquare$}
\end{example}

\noindent In Example~\ref{ex:first}, the facts created by chasing $\sigma_1$ do trigger $\eta_1$ and labelled nulls, corresponding to different sets, can be unified by $\eta_1$, whenever the two sets are discovered to be the same. This scenario is already beyond the expressive power of separability. Let us now analyse a different formulation of the same problem.

\begin{example}
\label{ex:second}
\begin{align*}
  \textit{element}(x) \rightarrow \exists z~\textit{comp}(x,z) \tag{$\sigma_1$} \\
  \textit{att}(x,k),\textit{att}(y,k) \to \exists z~\textit{comp}(x,z),  \textit{comp}(y,z)\tag{$\sigma_2$}\\
  \textit{comp}(x,z^{\prime}), \textit{comp}(x,z^{\prime\prime}) \to z^{\prime} = z^{\prime\prime} \tag{$\eta_1$}\\
\end{align*}
\textit{Every element is in a set ($\sigma_1$) and, whenever $x$ and $y$ have the same value $k$ for their attribute, there exists a set $z$ containing both of them ($\sigma_2$). Moreover, every element is in exactly one set ($\eta_1$).~$\blacksquare$}
\end{example}

Examples~\ref{ex:first} and~\ref{ex:second} are semantically equivalent, but Example~\ref{ex:second} uses only a key constraint. It is easy to see that $\eta_1$ unifies the facts produced while chasing $\sigma_1$ and $\sigma_2$. Therefore, also in this case, there is no separability. Let us validate this intuition with an example. Given the database $D=\{\textit{att}(1,A),\textit{att}(2,A),\textit{att}(3,A),\textit{element}(1),\textit{element}(2),\textit{element}(3)\}$,
and $Q = \textit{comp}(1,z),\textit{comp}(2,z),\textit{comp}(3,z)$ we have that in both the formulations
$\textit{chase}(D,\Sigma) \models Q$ and
$\textit{chase}(D, \Sigma_T) \not\models Q$ and thus we conclude $\Sigma$ is not separable. In particular, for Example~\ref{ex:first} from $\textit{chase}(D, \Sigma_T)$, we obtain $\{\textit{comp}(1,z_1),\textit{comp}(2,z_1),$\\$\textit{comp}(3,z_3)\}$ and for Example~\ref{ex:second}, we obtain $\{\textit{comp}(1,z_1),\textit{comp}(2,z_1),\textit{comp}(2,z_2),\textit{comp}(3,z_2),
\textit{comp}(1,z_3),\textit{comp}(3,z_3)\}$ and neither of them satisfies $Q$. 
On the other hand, in both the examples, $\textit{chase}(D,\Sigma)$ unifies $z_1,z_2$, and $z_3$ therefore satisfying $Q$. Moreover, it goes without saying that our EGDs are not innocuous as well, as we have just shown how they affect the result of query answering.

\medskip
Despite being not separable, the interaction between EGDs and TGDs in Examples~\ref{ex:first} and~\ref{ex:second} is somehow peculiar and lends itself to a simplified evaluation:
we first generate all the facts from the TGDs and then unify the generated labelled nulls.
In both cases, although EGDs are triggered by facts generated by TGDs, they do not trigger, in turn, TGDs. They do not exhibit \textit{forward interference} with other rules and only apply a final unification of the obtained labelled nulls, producing a less general but more fitting representation of the domain of interest.
This suggests that EGDs in our examples are to some extent \textit{harmless}, in that they do not affect or interfere with the application of other rules.

\subsection{Harmless EGDs}
The developed considerations are conducive to a new class of EGDs, not separable and able to express more complex reasoning problems possibly without altering the computational complexity of query answering in broad fragments of Datalog$^\pm$.
Intuitively, the idea is to consider a set of EGDs $\Sigma_E$ to be \textit{harmless} w.r.t.\ a specific set of rules $\Sigma$, if $\textit{chase}(\Sigma_T,D)$ produces
a more general result than  $\textit{chase}(\Sigma,D)$ for every $D$. This means that the result of the chase on $\Sigma_T$ can be always mapped onto the result of the chase of $\Sigma$.

\begin{definition}
\label{def:harmless}
\textit{Given a set of TGDs and EGDs $\Sigma = \Sigma_T \cup \Sigma_E$ over a schema $\mathbf{S}$, we define $\Sigma_E$ as \emph{harmless} if for every database $D$ over $\mathbf{S}$, (i) if $\textit{chase}(D,\Sigma)$ fails, then $\textit{chase}(D,\Sigma_T)$ does not satisfy $\Sigma_E$; (ii) if $\textit{chase}(D,\Sigma)$ does not fail, there exists a homomorphism $h$ mapping $\textit{chase}(D,\Sigma_T)$ onto $\textit{chase}(D,\Sigma)$.}
\end{definition}

In other words, harmless EGDs individuate the cases where a universal model for $D$ and $\Sigma$ can be obtained as a special case (i.e., an image of) of a universal model for $D$ and $\Sigma_T$. This notion can be equivalently formulated in terms of query answering.

\begin{definition}
\label{def:harmless_and_qa}
\textit{Given a set of TGDs and EGDs $\Sigma = \Sigma_T \cup \Sigma_E$ over a schema $\mathbf{S}$, we define $\Sigma_E$ as harmless if for every database $D$ over $\mathbf{S}$, (i) if $\textit{chase}(D,\Sigma)$ fails, then $\textit{chase}(D,\Sigma_T)$ does not satisfy $\Sigma_E$; (ii) if $\textit{chase}(D,\Sigma)$ does not fail, for every BCQ $Q$ over $\mathbf{S}$ there exists a partial homomorphism $h$ mapping $\textit{chase}(T,\Sigma_T)$ to $\textit{chase}(T,\Sigma)$, such that $\textit{chase}(D,\Sigma) \models Q$ iff $h(\textit{chase}(D,\Sigma_T)) \models Q$.}
\end{definition}

The reason we give two definitions is that we often use the simpler Definition~\ref{def:harmless}, while the style of Definition~\ref{def:harmless_and_qa} is often used in related worked, and is convenient for comparison to separable EGDs, as we show next.

\begin{theorem}
\label{th:def_equivalence}
Definitions~\ref{def:harmless} and~\ref{def:harmless_and_qa} are equivalent. 
\begin{proof}
\textit{Definitions~\ref{def:harmless}(i) and Definition~\ref{def:harmless_and_qa}(i) syntactically coincide, so we can concentrate on proving the equivalence of points (ii). We first show that Definition~\ref{def:harmless} implies Definition~\ref{def:harmless_and_qa}, and then that Definition~\ref{def:harmless_and_qa} implies Definition~\ref{def:harmless}.}

\smallskip\noindent
\textit{(Def~\ref{def:harmless} $\Rightarrow$ Def~\ref{def:harmless_and_qa})} By Def~\ref{def:harmless}, as $\Sigma_E$ is harmless, there exists a homomorphism $h$ mapping $\textit{chase}(D,\Sigma_T)$ onto $\textit{chase}(D,\Sigma)$. To prove Def~\ref{def:harmless_and_qa}, let us separately show that for every BCQ $Q$, there exists a partial homomorphism $h^\prime$  mapping $\textit{chase}(T,\Sigma_T)$ to $\textit{chase}(T,\Sigma)$ such that $\textit{chase}(D,\Sigma) \models Q$ $\Rightarrow$ $h^\prime(\textit{chase}(D,\Sigma_T)) \models Q$ and $\textit{chase}(D,\Sigma) \models Q$ $\Leftarrow$ $h^\prime(\textit{chase}(D,\Sigma_T)) \models Q$.\\
\noindent ($\Rightarrow$) Let us choose $h^\prime = h$. By Def~\ref{def:harmless}, $\textit{chase}(D,\Sigma)$ coincides with $h^\prime(\textit{chase}(D,\Sigma_T))$. Then, since by Def~\ref{def:harmless_and_qa} we have that $\textit{chase}(D,\Sigma)\models Q$, it follows that $h^\prime(\textit{chase}(D,\Sigma_T))\models Q$.\\
\noindent ($\Leftarrow$) Let us choose $h^\prime = h$. By Def~\ref{def:harmless}, $\textit{chase}(D,\Sigma)$ coincides with $h^\prime(\textit{chase}(D,\Sigma_T))$. Then, since by Def~\ref{def:harmless_and_qa} we have that $h^\prime(\textit{chase}(D,\Sigma_T))\models Q$, it follows that $\textit{chase}(D,\Sigma)\models Q$.

\smallskip\noindent
\textit{(Def~\ref{def:harmless_and_qa} $\Rightarrow$ Def~\ref{def:harmless}) Let us construct a BCQ $Q$ as the join of all the facts of $\textit{chase}(D,\Sigma)$.
By Def~\ref{def:harmless_and_qa}, there exists a partial homomorphism $h$ mapping $\textit{chase}(D,\Sigma_T)$ to $\textit{chase}(D,\Sigma)$, such that $\textit{chase}(D,\Sigma)\models Q$ iff $h(\textit{chase}(D,\Sigma_T))\models Q$. To prove Def~\ref{def:harmless}, we need to show that, for the homomorphism $h$ for such a query $Q$, we have that  $\textit{chase}(D,\Sigma)$ and $h(\textit{chase}(D,\Sigma_T))$ coincide. To this end, we need to show that $h$ is surjective (it maps $\textit{chase}(D,\Sigma_T)$ onto $\textit{chase}(D,\Sigma)$) and total (it maps all the elements of $\textit{chase}(D,\Sigma_T)$). Let us proceed separately for the two properties.}

\smallskip\noindent\underline{Surjectivity}.
By Def~\ref{def:harmless_and_qa}(ii-$\Rightarrow$) and given that $\textit{chase}(D,\Sigma)\models Q$ by construction, it follows that $h(\textit{chase}(D,\Sigma_T))\models Q$.
Then for every atom $\boldsymbol\varphi(\mathbf{x}) \in Q$, it holds $\boldsymbol\varphi(\mathbf{x}) \in h(\textit{chase}(D,\Sigma_T))$. Therefore $h$ is surjective.

\smallskip\noindent\underline{Totality}. Let us proceed by contradiction and assume $h$ is partial. Let $Q^\prime$ contain the facts of $\textit{chase}(D,\Sigma_T)$ that are not mapped by $h$. By Def~\ref{def:harmless_and_qa}(ii), there exists a homomorphism $h^\prime$ such that $h^\prime(\textit{chase}(D,\Sigma_T))\models Q^\prime$ iff  $\textit{chase}(D,\Sigma)\models Q^\prime$. Now, as $\textit{chase}(D,\Sigma_T)\models Q^\prime$ by construction, by Def~\ref{def:harmless_and_qa}(ii-$\Leftarrow$), it follows that $\textit{chase}(D,\Sigma)\models Q^\prime$. Since we have proved $h$ is surjective, then it covers the facts of $Q^\prime$ in $\textit{chase}(D,\Sigma)$. The only way for $h$ to cover the facts of $Q^\prime$ in $\textit{chase}(D,\Sigma)$ is to map them from the ones of $Q^\prime$ in $\textit{chase}(D,\Sigma_T)$, as $h$ is functional. This leads a contradiction, since the facts of $Q^\prime$ have been excluded from the domain of $h$. Therefore $h$ cannot be partial and is then total.
\end{proof}
\end{theorem}


\medskip
Let us now come to the relationship between harmless EGDs and separable EGDs, first showing the theoretical result and then analysing its implications with an example.

\begin{theorem}
\label{th:separability_and_harmless}
If a set of TGDs and EGDs $\Sigma = \Sigma_T \cup \Sigma_E$ is separable then $\Sigma_E$ is harmless (and not vice versa).
\begin{proof} We first show that EGD separability implies harmlessness.
By definition of separable EGDs, if $\textit{chase}(D,\Sigma)$ fails, then $D$ violates $\Sigma_E$, from which we conclude $\textit{chase}(D,\Sigma_T)$ violates $\Sigma_E$, i.e., Definition~\ref{def:harmless_and_qa}(i). Then, by separability, for every BCQ $Q$, it holds $\textit{chase}(D,\Sigma) \models Q$ iff $\textit{chase}(D,\Sigma_T) \models Q$. It is sufficient to choose the identity homomorphism for $h$ to derive Definition~\ref{def:harmless_and_qa}(ii) and conclude $\Sigma_E$ is harmless. Example~\ref{ex:second} is sufficient to show the reverse does not hold as, in such case, $\Sigma$ is not separable, while $\Sigma_E$ is harmless.
\end{proof}
\end{theorem}

It is also interesting to note that unlike separable EGDs, harmless EGDs do not limit the possible causes of chase failure to inherent violations of $\Sigma_E$ by $D$. On the contrary, in harmless EGDs, a violation of $\Sigma_E$ may even originate from facts generated by the TGDs. The following example (borrowed from~\cite{CaGP12}) highlights this aspect.

\begin{example}
\label{ex:fifth}

\begin{align*}
  \textit{r}(x,y) \rightarrow \exists z~\textit{s}(x,z,z) \tag{$\sigma_1$} \\
    \textit{s}(x,y,z), \textit{s}(x,y^{\prime},z^{\prime})
        \to y = y^{\prime},
            z = z^{\prime} \tag{$\eta_1$}\\
\end{align*}
\textit{Consider the database $D=\{r(a,b),s(a,b,c)\}$. The fact $s(a,z_1,z_1)$ is generated and it eventually violates $\eta_1$, which first unifies $z_1$ with $b$ and then unifies $z_1$ with $c$, causing a failure. The set of rules is not separable, yet the EGD is harmless since for every $D$ such that $\textit{chase}(D,\{\sigma_1\})$ does not violate $\eta_1$, a homomorphism to $\textit{chase}(D,\Sigma)$ can be easily found.~$\blacksquare$}

\end{example}

In operational terms, while separability assumes EGDs are pre-validated against $D$ and never applied during the chase, harmless EGDs allow the application of arbitrary EGD chase steps during the chase and the generated facts can even trigger TGD chase steps. The intertwined activation of harmless EGDs is tolerated since they do not give rise to a more general solution than the one that would be obtained with the only TGDs. However, this suggests a different chase execution technique, in which we first apply all the TGD steps (and thus generate $\textit{chase}(D,\Sigma_T)$)  and then all the EGD steps (thus implicitly define the homomorphism to $\textit{chase}(D,\Sigma)$), as we shall see.

\subsection{Undecidability Result}
\label{sec:decidability}

Before analysing sufficient syntactic conditions, we conclude the section by studying the general problem of determining the harmlessness of a set of EGDs. It turns out this problem is undecidable. We prove this with a reduction from Boolean conjunctive query answering, with the following strategy (inspired by~\cite{CGOP12}): given a BCQ Q, a database $D$ and a set of TGDs and EGDs $\Sigma$, we combine $D\cup\Sigma$ with a set $\Sigma^\prime$ of TGDs and non-harmless EGDs in such a way that for every database $D^\prime$ the TGDs of $\Sigma^\prime$ are triggered during the construction of $\textit{chase}(D^\prime,\Sigma^\prime)$, as the sole consequence of the triggering of non-harmful EGDs, iff $\textit{chase}(D,\Sigma)\models Q$. Intuitively, the reduction shows that deciding whether a BCQ Q is a model for $D\cup\Sigma$ is the same as deciding on the harmlessness of the EGDs in $\Sigma^\prime$.

\begin{theorem}
\label{th:undecidability}
Given a set of TGDs and EGDs $\Sigma = \Sigma_T \cup \Sigma_E$, the problem of deciding whether $\Sigma_E$ is harmless is undecidable.
\end{theorem}
\begin{proof}
We proceed by reduction from Boolean conjunctive query answering under arbitrary TGDs, known to be undecidable even for simple atomic queries~\cite{CaVa81}, to the co-problem of harmlessness. 
Let $Q: q \leftarrow r(c_1,c_2)$ be an arbitrary atomic query, $D$ a database  and $\Sigma$  a set of rules such that $r(c_1,c_2) \in \textit{chase}(D,\Sigma)$ and $r(c_1,c_2) \not\in D$. We construct a new set of rules $\Sigma^\prime$ starting from $\Sigma$ and adding to it: a rule $\top \to f$ for each fact $f \in D$ (where $\top$ is the constant \textit{true}) plus the following TGDs (over predicates $r^*$ and $s^*$, not appearing in $\Sigma$): $\sigma_1: r(c_1,c_2), r^*(x,x) \to s^*(x,x)$ and $\sigma_2: r(c_1,c_2),s^*(x,y)\to \exists z~r^*(x,z)$ and the EGD $\eta : s^*(x,y),r^*(x,z) \to x=z$. 
Now, we argue that $\textit{chase}(D,\Sigma) \models Q$ iff $\Sigma_E^\prime$ is not harmless. 

\medskip
\noindent($\Rightarrow$) Assume that $D\cup\Sigma \models Q$. We have to show that $\Sigma'_E$ is not harmless. Consider the database $D^\prime = \{s^*(1,2)\}$.
It is clear that $D^\prime \models \eta$ and $\textit{chase}(D,\Sigma^\prime)$ does not fail by construction. 
Given that $D\cup\Sigma \models Q$ holds, by assumption, we also know that $r(c_1,c_2) \in \textit{chase}(D,\Sigma)$. Hence $r(c_1,c_2) \in \textit{chase}(D^\prime,\Sigma)$.
Thus $\sigma_2$ can fire in $\textit{chase}(D^\prime,\Sigma^\prime)$, it activates $\eta$ that, in turn, activates $\sigma_1$ and so finally $s^*(1,1) \in \textit{chase}(D^\prime,\Sigma^\prime)$. 
Now let us consider $\textit{chase}(D^\prime,\Sigma^\prime \setminus \{\eta\})$. We have that $s^*(1,1) \not\in \textit{chase}(D^\prime,\Sigma^\prime \setminus \{\eta\})$ and no facts for $s^*$ are originally appearing in $\Sigma$. Hence there is no homomorphism $h$ from $\textit{chase}(D^\prime,\Sigma^\prime \setminus \{\eta\})$ to $\textit{chase}(D^\prime,\Sigma^\prime)$, and $\Sigma_E^\prime$ is therefore not harmless.

\medskip
\noindent($\Leftarrow$) If $\Sigma_E^\prime$ is not harmless, it is easy to verify that $r(c_1,c_2) \in \textit{chase}(D,\Sigma)$: in fact, if it were not the case, $\sigma_1$ and $\sigma_2$ could not be activated, contradicting our hypothesis.
\end{proof}

\section{Harmless vs Harmful EGDs}
\label{sec:harmless_harmful}

The essence of harmless EGDs lies in the fact they should not trigger specific TGDs that produce more general facts than can be otherwise generated. This observation is taken to the extreme in the undecidability proof of Theorem~\ref{th:undecidability}, where non-harmless EGDs are constructed by making them produce the only facts capable of activating certain TGDs and thus giving rise to facts otherwise not in the chase.
Yet, even if a given EGD $\eta$ has this harmful property, how can we be sure that there is some database $D$ for which it is actually activated?
Theorem~\ref{th:undecidability} witnesses the impossibility of having this answer. We face the problem from a different perspective and study syntactic conditions sufficient to guarantee that, independently of $D$, all the EGDs of $\Sigma$ do not exhibit a harmful behaviour. This is the goal of this section. 

\subsection{Harmful EGDs}
Before delving into the technical results, we present two examples where the EGDs, instead, exhibit a harmful behaviour. This will be important to understand what limitations should be posed by our syntactic condition.

\begin{example}
\label{ex:third}
\begin{align*}
  \textit{element}(x) \rightarrow \exists z~\textit{comp}(x,z) \tag{$\sigma_1$} \\
  \textit{rest}(x,y) \to \exists z~\textit{comp}(x,z),  \textit{comp}(y,z)\tag{$\sigma_2$}\\
  \textit{comp}(x,z^{\prime}), \textit{comp}(x,z^{\prime\prime}) \to z^{\prime} = z^{\prime\prime} \tag{$\eta_1$}\\
  \textit{comp}(x,z), \textit{comp}(y,z) \to \textit{siblings}(x,y) \tag{$\sigma_3$}\\
\end{align*}
\textit{Here, elements $x$ and $y$ are clustered together in a set $z$ if a binary relation \textit{rest} holds for the pair $\langle x,y \rangle$ ($\sigma_2$). Each element belongs to one set only ($\eta_1$). We then mark as siblings all the pairs of elements in the same set $z$ ($\sigma_3$). $\blacksquare$}
\end{example}

Given a database $D=\{\textit{element}(a),\textit{element}(b),\textit{element}(c),\textit{rest}(a,b),\textit{rest}(b,c)\}$, if we do not consider $\eta_1$, the only triggers for $\sigma_3$ would be the facts produced by $\sigma_2$, i.e., the pairs $\langle x,y \rangle$ for which \textit{rest} is in $D$. For each such pair, a new set $z$ would be created by $\sigma_2$, thus triggering the creation of \textit{sibling} facts by $\sigma_3$, i.e., $\textit{siblings}(a,b)$ and $\textit{siblings}(b,c)$.
What we wish to model here with $\eta_1$ is a form of transitivity of the \textit{rest} relation: if $\textit{rest}(a,b),\textit{rest}(b,c)$ holds, we would like $a,b,c$ to be all grouped within the same set. The EGD encodes this behaviour by unifying the sets of elements to which an element $x$ belongs. So actually, the EGD produces facts that act as triggers for $\sigma_3$, which yields the fact $\textit{siblings}(a,c)$ that could not be produced otherwise. A homomorphism as defined in Definition~\ref{def:harmless} is clearly impossible to find.

\begin{example}
\label{ex:fourth}
\begin{align*}
  \textit{a}(x,y,w) \rightarrow \exists z~\textit{b}(x,y,z) \tag{$\sigma_1$} \\
  \textit{a}(x,y,w) \rightarrow \exists z~\textit{b}(x,z,w) \tag{$\sigma_2$} \\
  \textit{b}(x,y^{\prime},w^{\prime}),
    \textit{b}(x,y^{\prime\prime},w^{\prime\prime})
        \to y^{\prime} = y^{\prime\prime},
            w^{\prime} = w^{\prime\prime} \tag{$\eta_1$}\\
  \textit{b}(x,y,z) \rightarrow \textit{f}(x,y,z) \tag{$\sigma_3$}\\
    \textit{b}(x,x,x) \rightarrow \textit{f}(x,x,x) \tag{$\sigma_4$}\\
\end{align*}
\textit{A data fusion case where the source $a$ is artificially split ($\sigma_1$ and $\sigma_2$) and then unified by $\eta_1$ and copied to $f$.~$\blacksquare$}
\end{example}

With reference to Example~\ref{ex:fourth}, consider $D=\{\textit{a}(1,1,1)\}$ and the BCQ $Q: q \leftarrow f(1,1,1)$. We have that $\textit{chase}(D,\Sigma) \models Q$, as it contains $f(1,1,1)$, generated by either $\sigma_3$ or $\sigma_4$. The chase $\textit{chase}(D,\Sigma_T)$ contains $f(1,1,z_1)$ and $f(1,z_2,1)$, generated by $\sigma_3$ and thus, with a homomorphism $h = \{z_1\to 1,~z_2\to 1\}$ from $\textit{chase}(D,\Sigma_T)$ to $\textit{chase}(D,\Sigma)$, we can satisfy $Q$. It easy to see that unless $D$ violates $\eta_1$ ($x$ is not a key), a homomorphism can be found for every $D$. Therefore $\Sigma_E$ is harmless. Nevertheless, $\eta_1$ exhibits a potentially harmful behaviour, since it generates the only triggering facts for $\sigma_4$. Fortunately, the facts generated by $\sigma_4$ are by construction less general than (i.e., an image of) those generated by $\sigma_3$. In fact, if we removed $\sigma_3$, then $\textit{chase}(D,\Sigma_T)$ would not contain any facts for $f$ and therefore no homomorphism $h$ would exist. In this case, $\Sigma_E$ would be harmful.
While we cannot check that potentially harmful EGDs like $\sigma_4$ are actually generating facts that could not be obtained otherwise, we would like a syntactic condition to frame out all these cases and to be sufficient to witness harmlessness.
Example~\ref{ex:fourth} inspires an interesting intuition concerning the relationship between labelled nulls and the variables unified in the EGDs. For harmlessness, we can tolerate that atom terms affected by an EGD appear in the body of a TGD (e.g., the second and the third term of $b$ appearing in the body of $\sigma_3$ and $\sigma_4$), yet those terms should not be determinant for the applicability of a TGD. It is the case of $\sigma_3$, whilst this does not happen in $\sigma_4$, where the values of affected terms are used to decide whether the TGD must fire.

\subsection{Safe Taintedness: A Syntactic Condition}
\label{sec:safe_taintedness}

Towards a formalization of a sufficient syntactic condition, we introduce a novel characterization of the rule terms that are directly or indirectly affected by EGDs. First we need to briefly recall some working definitions~\cite{GoPi15}.

\medskip\noindent\textbf{Working Definitions}. Let $p[i]$ denote the term appearing in the $i$-th position of atom $p$ and refer to it as \textit{position}. Let $\textit{head}(\sigma)$ and $\textit{body}(\sigma)$ be the head and the body of a rule $\sigma$, respectively; we also define $\textit{exist}(\sigma)$ the set of existentially quantified variables of $\sigma$.
Given a set of rules $\Sigma$, a position $p[i]$ is inductively defined as \textit{affected position} if: (i) for some TGD $\sigma \in \Sigma$ s.t.\ and some variable $v \in \textit{exist}(\sigma)$, $v$ appears in position $p[i]$;
    (ii) for some TGD $\sigma$ and some variable $v \in \textit{body}(\sigma) \cap \textit{head}(\sigma)$, $v$ appears only in affected positions in $\textit{body}(\sigma)$ (i.e., it is a \textit{harmful variable}) and in position $p[i]$ in $\textit{head}(\sigma)$ (i.e., it is a \textit{dangerous variable}).

\medskip
\noindent\textbf{Taintedness}. We are now ready to introduce the notions of \emph{tainted position} and \emph{tainted variable}.

\begin{definition}
\label{def:tainted}
\textit{Given a set of rules $\Sigma$, a position $p[i]$ is inductively defined as \textit{tainted position} if:}
(i) \textit{$p[i]$ is affected and $\Sigma$ contains an EGD $\eta = \boldsymbol{\phi}(\mathbf{x}) \to x_i = x_j$ such that $p[i]$ is the position of $x_i$ (resp.\ $x_j$) in $\textit{body}(\eta)$ and $x_i$ (resp.\ $x_j$) is harmful in $\eta$;} (ii) \textit{for some TGD $\sigma \in \Sigma$ and some variable $v \in \textit{body}(\sigma) \cap \textit{head}(\sigma)$, $v$ appears in a tainted position in $\textit{body}(\sigma)$ (resp. $\textit{head}(\sigma)$) and in position $p[i]$ in $\textit{head}(\sigma)$ (resp. $\textit{body}(\sigma)$).}
\end{definition}

\begin{definition}
\textit{A variable appearing in a tainted position of a rule $\sigma$ is named \textit{tainted variable} (with respect to $\sigma$). Let $\textit{tainted}(\sigma)$ be the set of all the tainted variables in the body of $\sigma$.}
\end{definition}

Observe that only affected positions can be ``upgraded'' to tainted and that, unlike affected variables that propagate to the head only when harmful (so all their respective body positions are affected), for tainted variables to propagate just a propagating tainted variable suffices. Also, unlike affected variables, they propagate from the body to the head and vice versa.
Moreover, notice that with respect to Definition~\ref{def:tainted}, if $x_i$ of $\eta$ does not appear in an affected position in $\textit{body}(\eta)$, only the position of $x_j$ will be tainted. On the contrary, if both $x_i$ and $x_j$ appear in affected positions of $\textit{body}(\eta)$, both the positions will become tainted.
The rationale here is that the definition tracks the positions in which the specific value of the labelled null may be relevant for joins and selections.

\medskip
We are now ready to give a syntactic characterization of harmless EGDs, which we name \textit{safe taintedness}.

\begin{theorem}
\label{th:tainted_implies_harmless}
Let $\Sigma=\Sigma_T\cup\Sigma_E$ be a set of TGDs $\Sigma_T$ and EGDs $\Sigma_E$ over schema $\mathbf{S}$. The set $\Sigma_E$ is harmless if (safe taintedness) for every %
dependency $\sigma \in \Sigma$: (i) every variable $v\in \textit{tainted}(\sigma)$  appears only once in $\textit{body}(\sigma)$, and, (ii) there are no constants appearing in tainted positions (i.e., there are no ground tainted positions).

\begin{proof}
Assume that (safe taintedness) for every %
dependency $\sigma \in \Sigma$: (i) every variable $v\in \textit{tainted}(\sigma)$  appears only once in $\textit{body}(\sigma)$, and, (ii) there are no constants appearing in tainted positions (i.e., there are no ground tainted positions). We need to show that $\Sigma$ is harmless.
By Definition~\ref{def:harmless}, we thus need to prove that for every database $D$ of schema $\mathbf{S}$, (i)~if 
$\textit{chase}(D,\Sigma)$ fails, then $\textit{chase}(D,\Sigma_T)$ does not satisfy $\Sigma_E$; (ii)~if $\textit{chase}(D,\Sigma)$ does not fail, there exists a homomorphism mapping  $\textit{chase}(D,\Sigma_T)$ onto $\textit{chase}(D,\Sigma)$.

\smallskip
\noindent
\textbf{No tainted positions}.
Let us first analyse the trivial cases.
If $\Sigma_E = \emptyset$, the chase cannot fail by construction, proving~Definition~\ref{def:harmless}(i), and it is sufficient to choose the empty homomorphism $h=\emptyset$ to fulfil ~Definition~\ref{def:harmless}(ii). Therefore $\Sigma_E$ is harmless. 
If $\Sigma_E \neq \emptyset$ but there are no tainted positions in $\Sigma$, by Definition~\ref{def:tainted}, either there are no affected positions in $\Sigma$, or there are no harmful variables in the body of any EGD $\eta: \boldsymbol{\phi}(\mathbf{x}) \to x_i=x_j$ of $\Sigma_E$. In either case, in the chase construction, EGDs only equate ground values. Therefore $\textit{chase}(D,\Sigma)$ can fail only for a hard violation of $\Sigma_E$ by $\textit{chase}(D,\Sigma_T)$, fulfilling Definition~\ref{def:harmless}(i), and since no EGD in $\Sigma_E$ binds any labelled null to a ground value, it is sufficient to choose the empty homomorphism $h=\emptyset$ to fulfil Definition~\ref{def:harmless}(ii).

\smallskip
\noindent
\textbf{Tainted positions}.
Let us now account for the presence of tainted positions in TGDs $\sigma: \boldsymbol{\phi}(\mathbf{x}) \to \exists\mathbf{z}~\boldsymbol{\psi}(\mathbf{y},\mathbf{z})$ of $\Sigma_T$. 
We start by proving that Definition~\ref{def:harmless}(ii) holds.
We must show that if $\textit{chase}(D,\Sigma)$ does not fail, there exists a homomorphism $h$ mapping $\textit{chase}(D,\Sigma_T)$ onto $\textit{chase}(D,\Sigma)$, 
and so that $\textit{chase}(D,\Sigma) = h(\textit{chase}(D,\Sigma_T))$. We proceed by proving the two directions separately, i.e., cases 
(a) $\textit{chase}(D,\Sigma) \subseteq h(\textit{chase}(D,\Sigma_T))$, and (b) $h(\textit{chase}(D,\Sigma_T)) \subseteq \textit{chase}(D,\Sigma)$.

\smallskip
\noindent
\textbf{Case (a)}:  $\textit{chase}(D,\Sigma) \subseteq h(\textit{chase}(D,\Sigma_T))$ holds.
To prove (a), we proceed by induction on $\textit{chase}(D,\Sigma)$
and show that it is possible to incrementally build $h$ 
such that it holds as an induction invariant that 
for every conjunction of facts $\boldsymbol{\varPsi}(\mathbf{y},\mathbf{z}) \in \textit{chase}(\Sigma,D)$ there exists a conjunction of facts
$\boldsymbol{\varPsi^T}(\mathbf{y},\mathbf{z}) \in \textit{chase}(\Sigma_T,D)$ which is mapped onto $\boldsymbol{\varPsi}(\mathbf{y},\mathbf{z})$
by $h$.

\smallskip\noindent
(a.1)~\textit{Base case}. Let $\boldsymbol{\varPsi_1}(\mathbf{y},\mathbf{z})$ be a conjunction of facts produced by one single initial chase step $\boldsymbol{\varphi_0}(\mathbf{x}) \xrightarrow{\sigma_\theta} \boldsymbol{\varPsi_1}(\mathbf{y},\mathbf{z})$ of $\textit{chase}(D,\Sigma)$, where
$\sigma_\theta$ denotes that a TGD $\sigma$ has been applied by the first TGD chase step, with a triggering homomorphism $\theta$ from $\mathbf{x}$ to $D$ as defined in Section~\ref{sec:preliminaries}. 
Since $\boldsymbol{\varphi_0}(\mathbf{x})$ is in $D$, the chase step $\boldsymbol{\varphi_0}(\mathbf{x}) \xrightarrow{\sigma_\theta} \boldsymbol{\varPsi_1^T}(\mathbf{y},\mathbf{z})$ of $\textit{chase}(D,\Sigma_T)$ is activated and hence it holds that  $\boldsymbol{\varPsi_1}(\mathbf{y},\mathbf{z}) = \boldsymbol{\varPsi_1^T}(\mathbf{y},\mathbf{z})$ by construction. We initialize $h_0=\emptyset$.

\smallskip\noindent
(a.2)~\textit{Inductive case}. 
We now consider a conjunction of facts $\boldsymbol{\varPsi_n}(\mathbf{y},\mathbf{z}) \in chase^{n}(D,\Sigma)$ generated by the first $n$ chase steps,
and want to prove
the existence of a
conjunction of facts $\boldsymbol{\varPsi^T_n}(\mathbf{y},\mathbf{z}) \in \textit{chase}(D,\Sigma_T)$ such that $h_n(\boldsymbol{\varPsi^T_n}(\mathbf{y},\mathbf{z}))=\boldsymbol{\varPsi_n}(\mathbf{y},\mathbf{z})$.
Two cases are possible. 
\begin{itemize}[leftmargin=4mm]

\item \textit{EGD chase step}. $\boldsymbol{\varPsi_n}(\mathbf{y},\mathbf{z})$ has been generated as a consequence of the bindings induced in $\textit{chase}(D,\Sigma)$ by an \textit{EGD chase step}: $\boldsymbol{\varphi_{n-1}}(\mathbf{x}) \xrightarrow{\sigma_{\theta}} x_i=x_j$. By inductive hypothesis, we assume that for $\boldsymbol{\varPsi_{n-1}}(\mathbf{y},\mathbf{z}) \in \textit{chase}(D,\Sigma)$ there is a conjunction of facts $\boldsymbol{\varPsi^T_{n-1}}(\mathbf{y},\mathbf{z}) \in \textit{chase}(D,\Sigma_T)$ such that $\boldsymbol{\varPsi_{n-1}}(\mathbf{y},\mathbf{z}) = h_{n-1}(\boldsymbol{\varPsi^T_{n-1}}(\mathbf{y},\mathbf{z}))$. 
We now let $h_n = h_{n-1} \circ \{x_i\to x_j\}$, that is, we update $h$, by composing it with the assignments introduced by the EGD chase step. By construction, it follows $h_n(\boldsymbol{\psi^T_n}(\mathbf{y},\mathbf{z}))=\boldsymbol{\psi_n}(\mathbf{y},\mathbf{z})$.

\item \textit{TGD chase step}. $\boldsymbol{\varPsi_n}(\mathbf{y},\mathbf{z})$ has been generated by a \textit{TGD chase step} $\boldsymbol{\varphi_{n-1}}(\mathbf{x}) \xrightarrow{\sigma_{\theta}} \boldsymbol{\varPsi_n}(\mathbf{y},\mathbf{z})$.
By inductive hypothesis, we assume that for $\boldsymbol{\varphi_{n-1}}(\mathbf{x}) \in \textit{chase}(D,\Sigma)$ there is a conjunction of facts $\boldsymbol{\varphi^T_{n-1}}(\mathbf{x}) \in \textit{chase}(D,\Sigma_T)$ such that $\boldsymbol{\varphi_{n-1}}(\mathbf{x}) = h_{n-1}(\boldsymbol{\varphi^T_{n-1}}(\mathbf{x}))$.
We show that if $\theta(\boldsymbol{\phi}(\mathbf{x})) \subseteq \textit{chase}^{n-1}(D,\Sigma)$, it follows $\theta(\boldsymbol{\phi}(\mathbf{x})) \subseteq \textit{chase}^{n-1}(D,\Sigma_T)$.
To prove this, we argue that if $\theta$ maps $\boldsymbol{\phi}(\mathbf{x})$ to $\boldsymbol{\varphi_{n-1}}(\mathbf{x})$, then it also maps
$\boldsymbol{\phi}(\mathbf{x})$ to
$\boldsymbol{\varphi^T_{n-1}}(\mathbf{x})$. As, $\boldsymbol{\varphi_{n-1}}(\mathbf{x}) = h_{n-1}(\boldsymbol{\varphi^T_{n-1}}(\mathbf{x}))$, the only two possibilities that $\theta$ applies to $\boldsymbol{\varphi_{n-1}}(\mathbf{x})$ and not to $\boldsymbol{\varphi^T_{n-1}}(\mathbf{x})$ are that: 
(p1) $h_{n-1}$ replaces a labelled null $\nu_i$ in position $\boldsymbol{\phi}[i]$ of $\boldsymbol{\varphi_{n-1}^T}(\mathbf{x})$ with a constant
and $\mathbf{\boldsymbol{\phi}}[i]$ is ground in $\boldsymbol{\phi}(\mathbf{x})$; 
(p2) $h_{n-1}$ replaces a labelled null $\nu_i$ in position $\boldsymbol{\phi}[i]$ of $\boldsymbol{\varphi^T_{n-1}}(\mathbf{x})$ with another variable null $\nu_j$, already appearing in position $\boldsymbol{\phi}[j]$, with $i\neq j$, and $\boldsymbol{\phi}(\mathbf{x})$ contains the same variable in positions $\boldsymbol{\phi}[i]$ and $\boldsymbol{\phi}[j]$.\\

Possibilities (p1) and (p2) do not occur in any of our cases, which we show next. In the base case, $h=\emptyset$. In the TGD inductive case, $h$ is not extended with assignments to constants, nor is it extended with assignments to labelled nulls already appearing in any atom (as we specify at the end this point). Finally, let us focus on the EGD inductive case (see previous point). If $\boldsymbol{\phi}[i]$ is ground in $\boldsymbol{\phi}(\mathbf{x})$, by safe taintedness, $\boldsymbol{\phi}[i]$ cannot be tainted, therefore, no EGD chase step can map any labelled null appearing in $\boldsymbol{\phi}[i]$ of $\boldsymbol{\varphi}^T_{n-1}(\mathbf{x})$ into a constant, and thus $h$ cannot contain such an assignment. If $\boldsymbol{\phi}[i]$ is tainted and is not ground
in $\boldsymbol{\phi}(\mathbf{x})$, by safe taintedness, there cannot be any other positions $\boldsymbol{\phi}[j]$, with $i\neq j$ where  $\boldsymbol{\phi}(\mathbf{x})$ has the same variable. Therefore neither (p1) nor (p2) ever occur.\\

We now build $h_n$ by extending $h_{n-1}$ as $h_n=h_{n-1} \cup \{{z^T_i}\to {z_i},\ldots\}$, where $z^T\to z_i$ denotes the correspondences between fresh labelled nulls of $\mathbf{z}$ in
$\boldsymbol{\varPsi^T_n}(\mathbf{y},\mathbf{z})$ resp.\ $\boldsymbol{\varPsi}(\mathbf{y},\mathbf{z})$.
It follows that  $h_n(\boldsymbol{\varPsi^T_n}(\mathbf{y},\mathbf{z}))=\boldsymbol{\varPsi_n}(\mathbf{y},\mathbf{z})$.

\end{itemize}

\smallskip
\noindent
\textbf{Case (b)}:  $h(\textit{chase}(D,\Sigma_T)) \subseteq \textit{chase}(D,\Sigma)$ holds. 
To prove (b), we proceed by induction on $\textit{chase}(D,\Sigma_T)$
and show that it is possible to incrementally build a 
homomorphism $h$
such that it holds as an induction invariant that 
for every conjunction of facts $\boldsymbol{\varPsi}^T(\mathbf{y},\mathbf{z}) \in \textit{chase}(\Sigma_T,D)$ there exists a conjunction of facts
$\boldsymbol{\varPsi}(\mathbf{y},\mathbf{z}) \in \textit{chase}(\Sigma,D)$ onto which $\boldsymbol{\varPsi}^T(\mathbf{y},\mathbf{z})$ is mapped
by $h$. 

\smallskip\noindent
(b.1)~\textit{Base case}. Let $\boldsymbol{\varPsi^T_1}(\mathbf{y},\mathbf{z})$ be a conjunction of facts produced by one single initial chase step $\boldsymbol{\varphi^T_0}(\mathbf{x}) \xrightarrow{\sigma_\theta} \boldsymbol{\varPsi^T_1}(\mathbf{y},\mathbf{z})$ of $\textit{chase}(D,\Sigma_T)$. 
Since $\boldsymbol{\varphi^T_0}(\mathbf{x})$ is in $D$, the chase step $\boldsymbol{\varphi^T_0}(\mathbf{x}) \xrightarrow{\sigma_\theta} \boldsymbol{\varPsi_1}(\mathbf{y},\mathbf{z})$ of $\textit{chase}(D,\Sigma)$ is activated and hence it holds $\boldsymbol{\varPsi^T_1}(\mathbf{y},\mathbf{z}) = \boldsymbol{\varPsi_1}(\mathbf{y},\mathbf{z})$ by construction. We initialize $h_0=\emptyset$.

\smallskip\noindent
(b.2)~\textit{Inductive case}. 
We now consider a conjunction of facts $\boldsymbol{\varPsi^T_n}(\mathbf{y},\mathbf{z}) \in chase^{n}(D,\Sigma_T)$ generated by the first $n$ chase steps,
and want to prove
the existence of a
conjunction of facts $\boldsymbol{\varPsi_n}(\mathbf{y},\mathbf{z}) \in \textit{chase}(D,\Sigma)$ such that $\boldsymbol{\varPsi_n}(\mathbf{y},\mathbf{z})=h_n(\boldsymbol{\varPsi^T_n}(\mathbf{y},\mathbf{z}))$. For $\boldsymbol{\varPsi^T_n}(\mathbf{y},\mathbf{z})$ to exist, it must have been generated by a \textit{TGD chase step} $\boldsymbol{\varphi^T_{n-1}}(\mathbf{x}) \xrightarrow{\sigma_{\theta}} \boldsymbol{\varPsi^T_n}(\mathbf{y},\mathbf{z})$.
By inductive hypothesis, we assume that for $\boldsymbol{\varphi^T_{n-1}}(\mathbf{x}) \in \textit{chase}(D,\Sigma_T)$ there is a conjunction of facts $\boldsymbol{\varphi_{n-1}}(\mathbf{x}) \in \textit{chase}(D,\Sigma)$ such that $h_{n-1}(\boldsymbol{\varphi^T_{n-1}}(\mathbf{x})) =\boldsymbol{\varphi_{n-1}}(\mathbf{x})$. If $\theta(\boldsymbol{\phi}(\mathbf{x})) \subseteq \textit{chase}^{n-1}(D,\Sigma_T)$, there exists a triggering homomorphism $\theta^\prime = h_{n-1}\circ\theta$ such that $\theta^\prime(\boldsymbol{\phi}(\mathbf{x})) \subseteq \textit{chase}^{n-1}(D,\Sigma)$.
We now build $h_n$ by extending $h_{n-1}$ as $h_n=h_{n-1} \cup \{{z^T_i}\to {z_i},\ldots\}$, where $z^T\to z_i$ denotes the correspondences between fresh labelled nulls of $\mathbf{z}$ in
$\boldsymbol{\varPsi^T_n}(\mathbf{y},\mathbf{z})$ resp.\ $\boldsymbol{\varPsi}(\mathbf{y},\mathbf{z})$. It follows $\boldsymbol{\varPsi_n}(\mathbf{y},\mathbf{z})=h_n(\boldsymbol{\varPsi^T_n}(\mathbf{y},\mathbf{z}))$. 

\smallskip
\noindent
Having shown both cases (a) and (b), this concludes our proof that  ~\ref{def:harmless}(ii) holds.

\smallskip
\noindent
\textbf{Definition ~\ref{def:harmless}(i) holds}.
We conclude the proof by showing that safe taintedness also implies Definition~\ref{def:harmless}(i). Let us proceed by contradiction and assume $\textit{chase}(D,\Sigma)$ fails while $\textit{chase}(D,\Sigma_T) \models \Sigma_E$.
This means that there is a conjunction of facts $\boldsymbol{\varPsi} \in \textit{chase}(D,\Sigma)$ violating an EGD of $\Sigma_E$ and $\boldsymbol{\varPsi} \not \in \textit{chase}(D,\Sigma_T)$. For this to happen, there must be a TGD chase step $\boldsymbol{\varphi}(\mathbf{x}) \xrightarrow{\sigma_{\theta}} \boldsymbol{\varPsi}({\mathbf{y,z}})$ or an EGD chase step $\boldsymbol{\varphi}(\mathbf{x}) \xrightarrow{\sigma_{\theta}} x_i=x_j$ of $\textit{chase}(D,\Sigma)$ giving rise to $\boldsymbol{\varPsi}$ as a consequence of the binding $x_i=x_j$. So, in either cases $\theta(\boldsymbol{\phi}(\mathbf{x})) \subseteq \textit{chase}(D,\Sigma)$, while there is no homomorphism $\theta^\prime$ such that $\theta^\prime(\boldsymbol{\phi}(\mathbf{x})) \subseteq \textit{chase}(D,\Sigma_T)$. This means there is another EGD chase step $\boldsymbol{\varphi}(\mathbf{x}) \xrightarrow{\sigma_{\theta^{\prime\prime}}} x_i^\prime=x_j^\prime$ that maps $\boldsymbol{\varphi}(\mathbf{x})$ into $\boldsymbol{\varphi}^\prime(\mathbf{x})$ as an effect of the binding $x_i^\prime=x_j^\prime$. There are two possibilities for such $\theta^\prime$ not to exist: (p1) the binding $x_i=x_j$ replaces a labelled null $\nu_i$ in position $\boldsymbol{\phi}[i]$ of $\boldsymbol{\varphi}(\mathbf{x})$ with a constant
and $\mathbf{\boldsymbol{\phi}}[i]$ is ground in $\boldsymbol{\phi}(\mathbf{x})$; 
(p2) the binding $x_i=x_j$ replaces a labelled null $\nu_i$ in position $\boldsymbol{\phi}[i]$ of $\boldsymbol{\varphi}(\mathbf{x})$ with another variable null $\nu_j$, already appearing in position $\boldsymbol{\phi}[j]$, with $i\neq j$, and $\boldsymbol{\phi}(\mathbf{x})$ contains the same variable in positions $\boldsymbol{\phi}[i]$ and $\boldsymbol{\phi}[j]$. Option (p1) contradicts the safe taintedness hypothesis: as $\boldsymbol{\phi}[i]$ is tainted, it cannot be ground. Options (p2) contradicts the safe taintedness hypothesis as well: as as $\boldsymbol{\phi}[i]$ is tainted, there cannot be any other positions $\boldsymbol{\phi}[j]$, with $i\neq j$ where  $\boldsymbol{\phi}(\mathbf{x})$ has the same variable. Therefore neither (p1) nor (p2) ever take place. Therefore we can conclude that if $\textit{chase}(D,\Sigma)$ fails, then $\textit{chase}(D,\Sigma_T) \not\models \Sigma_E$.

\smallskip\noindent
Having proved Definition~\ref{def:harmless}(i) and~(ii), we can conclude that under the theorem hypotheses, $\Sigma_E$ is harmless.
\end{proof}
\end{theorem}

Let us now go back to the examples we have seen so far and put the safe taintedness syntactic condition into action. The EGDs of our first Example~\ref{ex:intro_example} are harmless. In fact, EGDs $\eta_1$ and $\eta_2$ taint the position $\textit{partOf}[2]$, which however never corresponds to either repeated or ground variables in the body of $\sigma_1$.
In Example~\ref{ex:first}, $\textit{comp}[2]$ is tainted, but this atom does not appear in the body of any TGDs, thus the EGD is harmless; similarly in Example~\ref{ex:second} and~\ref{ex:fifth}. In Example~\ref{ex:third}, $\textit{comp}[2]$ is tainted and $z$ in $\sigma_3$ is tainted and appears multiple times in the body, so we cannot conclude that $\eta_1$ is harmless. In Example~\ref{ex:fourth}, positions $b[2]$ and $b[3]$ are tainted. Tainted variables are not repeated in $\sigma_3$, while they are repeated in $\sigma_4$. Hence, also in this case, we cannot conclude $\eta_1$ is harmless.
We now propose one final example that highlights the importance of back-propagation of tainted positions.
\begin{example}
\label{ex:back_propagate}
\begin{align*}
  \textit{s}(x,y), a(k,y) \rightarrow q(x,k) \tag{$\sigma_1$} \\
  \textit{a}(x,y) \rightarrow \exists z~s(x,z) \tag{$\sigma_2$} \\
  \textit{s}(x,y) \rightarrow r(y,x) \tag{$\sigma_3$} \\
  \textit{r}(x,y),\textit{r}(x^\prime,y) \rightarrow x=x^\prime \tag{$\eta_1$} \\
\end{align*}
\textit{Let us consider a database $D=\{a(1,2), r(2,1)\}$. From $a(1,2)$ we generate $s(1,z_1)$ by $\sigma_2$ and then $r(z_1,1)$ by $\sigma_3$. The EGD $\eta_1$ then unifies $z_1=2$ and so $\sigma_1$ can fire giving rise to $q(1,1).$
$\blacksquare$}
\end{example}

Interestingly, in Example~\ref{ex:back_propagate}, 
$\eta_1$ taints position $r[1]$ which does not directly appear in the body of any rules. Conversely, it back propagates to $s[2]$, which in the body of $\sigma_1$ is duplicated, witnessing the impossibility to conclude $\eta_1$ is harmless and, in fact, exhibits a harmful behaviour for example with the database $D$ in the example.

\section{Conjunctive Query Answering with Harmless EGDs}
\label{sec:efficient_evaluation}

Definition~\ref{def:harmless} immediately suggests an approach to ontological BCQ answering in the presence of a set of dependencies $\Sigma = \Sigma_T \cup \Sigma_E$, where $\Sigma_E$ is a set of harmless EGDs. 

\begin{enumerate}[leftmargin=4mm]
    \item Decide weather $D\cup\Sigma$ is satisfiable, otherwise trivially conclude $D\cup\Sigma \models Q$.
    \item If $D\cup\Sigma$ is satisfiable, by Definition~\ref{def:harmless_and_qa}, there exists a homomorphism $h$ mapping 
    $\textit{chase}(D,\Sigma_T)$ onto $\textit{chase}(D,\Sigma)$
    such that $\textit{chase}(D,\Sigma) \models Q$ iff $h(\textit{chase}(D,\Sigma_T))\models Q$.
    Determine $h$ and check whether $h(\textit{chase}(D,\Sigma_T))\models Q$.
\end{enumerate}

\medskip\noindent\textbf{Warded Datalog$^\pm$}. To show the main results, we extend the working definitions given in Section~\ref{sec:safe_taintedness} by recalling the preliminaries of Warded Datalog$^\pm$~\cite{GoPi15, BeSG18}. Given a set of rules $\Sigma$, a rule $\sigma \in \Sigma$ is \textit{warded} if all the dangerous variables $v \in \textit{body}(\sigma)$ appear in a single body atom, the \textit{ward}. A set $\Sigma$ is warded if the body variables of all the rules in $\Sigma$ are warded. BCQ answering under a set of TGDs $\Sigma_T$ is decidable and \textsf{PTIME} in Warded Datalog$^\pm$. 

We say that two sets of rules $\Sigma$ and $\Sigma^\prime$ are \textit{CQ-equivalent} if, for every database $D$ and BCQ $Q$, we have that $D\cup\Sigma \models Q$ iff $D\cup\Sigma^\prime \models Q$. The theoretical underpinnings of Warded Datalog$^\pm$ allow to rewrite $\Sigma_T$ (with the \textit{harmful join elimination} procedure) into a CQ-equivalent set of TGDs $\Sigma^\prime$, for which query answering can be 
performed within a finite variant of the chase ($\textit{chase}^W$) where facts are considered equivalent modulo isomorphism of labelled nulls, and so facts isomorphic to facts already generated are not generated. We call such assumption ---and $\textit{chase}^W$--- \textit{warded semantics}. 
In particular, for a database $D$ and a rewriting $\Sigma_T^\prime$ obtained via harmful join elimination from a set of warded TGDs $\Sigma_T$, the following property holds: there exists a finite $\textit{chase}^W(D,\Sigma_T^\prime) \subseteq \textit{chase}(D,\Sigma_T^\prime)$ such that, for every atomic BCQ $Q$,  $Q\models \textit{chase}^W(D,\Sigma_T^\prime)$ iff $Q\models \textit{chase}(D,\Sigma_T^\prime)$. Considering atomic BCQ is not a limitation here, as any BCQ can be rewritten as an atomic BCQ, by encoding the conjunctive query as a new rule of $\Sigma_T$. While warded semantics applies to warded TGDs, no conclusions can be drawn on whether it is applicable in the context of harmless EGDs (and, as we shall see, it is not). Hence, we consider
the \textit{universal semantics} of $D$ and $\Sigma_T$ and represent it as a purely theoretical subset $\textit{chase}^B(D,\Sigma_T)\subseteq\textit{chase}(D,\Sigma_T)$, finite and CQ-equivalent to $\textit{chase}(D,\Sigma_T)$, i.e., a universal model for $D\cup \Sigma_T$.

\medskip
In the next sections, we analyse the described approach to BCQ answering in the case where TGDs are specified in Warded Datalog$^\pm$, because of its relevance as a fragment with a favourable trade-off between computational complexity and expressive power. In particular, we will provide a concrete space-efficient implementation of the universal semantics $\textit{chase}^B$ for Warded Datalog$^\pm$, namely $\textit{chase}^{B_W}(D,\Sigma)$, which keeps enough information to apply harmless EGDs.
While beyond the scope of this paper, we conjecture that the results generalize to any decidable Datalog$^\pm$ fragment and we believe that BCQ decidability and data complexity just descend from those of the fragment adopted for the TGDs, a very desirable property of harmless EGDs.

\medskip\noindent\textbf{Checking Satisfiability}
To check satisfiability we encode the assignments performed by the EGDs into a set of TGDs $\Sigma_V$. Then, given a set $\Sigma = \Sigma_T \cup \Sigma_E$, where $\Sigma_E$ is harmless, we have that satisfiability holds if and only if for a non-contradiction query $Q_V$, it holds $D\cup(\Sigma_T\cup\Sigma_V)\not\models Q_V$, as by Definition~\ref{def:harmless}(i), all the failures of $\textit{chase}(D,\Sigma)$ are hard violations of the chase of $\Sigma_T$ over $\Sigma_E$. The encoding is done as follows. We initialize $D^\prime = \textit{chase}^B(D,\Sigma_T)$, where the labelled nulls generated by the chase are considered as constants. We then extend $D^\prime$ 
with facts $\textit{neq}(c_1,c_2)$ resp.\ $\textit{eq}(c_1,c_1)$ for each pair of distinct resp.\ equal constants $c_1,c_2 \in \textit{dom}(D)$ (the set of all constants of $D$).
We define a new set $\Sigma_V$ containing a TGD $\sigma$ for each EGD $\eta: \boldsymbol{\phi}(\mathbf{x}) \to x_i=x_j$ of $\Sigma_E$. The TGD $\sigma$ is built as $\boldsymbol{\phi}^\prime(\mathbf{x})\to\textit{eq}(x_i,x_j)$.
We also add to $\Sigma_T^\prime$ the usual rules to encode equality symmetry and transitivity. 
Finally, the non-contradiction query is defined as $Q_V: q \leftarrow \textit{eq}(x_i,x_j),\textit{neq}(x_i,x_j)$.
$\Sigma_V$ does not contain existentials and is therefore warded.
It is not difficult to prove that if $D^\prime\cup\Sigma_V \models Q_V$ ---which can be tested by decidability of BCQ answering in Warded Datalog$^\pm$--- then $\textit{chase}(D,\Sigma)$ fails and $D\cup\Sigma$ is unsatisfiable.

\medskip\noindent\textbf{Finding the Homomorphism}. Having checked the satisfiability of $D\cup\Sigma$, a technique to determine a homomorphism from $\textit{chase}(D,\Sigma_T)$ onto $\textit{chase}(D,\Sigma)$ is as follows. We initialize an empty homomorphism $h=\emptyset$. We then apply the EGDs of $\Sigma_E$ on $\textit{chase}^B(D,\Sigma_T)$ to fixpoint, i.e., we compute $\textit{chase}^H(D,\Sigma) = \textit{chase}(\textit{chase}^B(D,\Sigma_T),\Sigma_E)$. 
This process is finite for a fixed set of EGDs~\cite{CGOP12}. 
For each EGD chase step $\boldsymbol{\varphi}(\mathbf{x}) \xrightarrow{\sigma_{\theta}} x_i=x_j$ that binds a labelled null $\theta(x_i)$ to either a constant or another labelled null $\theta(x_j)$, we update $h$ by adding the assignment ${\theta(x_i) \to \theta(x_j)}$,  and replacing all occurrences of $\theta(x_i)$ appearing in the right-hand side of any assignment in $h$ with $\theta(x_j)$.
Then, for the labelled nulls of each $\boldsymbol{\varphi}^\prime(\mathbf{x})$ of $\textit{chase}(D,\Sigma_T) \setminus \textit{chase}^B(D,\Sigma_T)$, isomorphic to a fact $\boldsymbol{\varphi}(\mathbf{x})$ of $\textit{chase}^B(D,\Sigma)$, we inherit the assignments from those enforced by $h$ on the respective labelled nulls of $\boldsymbol{\varphi}(\mathbf{x})$ and extend $h$ as a consequence.
Given the homomorphism $h$, we have that $D\cup\Sigma \models Q$ can be decided by checking whether $h(\textit{chase}(D,\Sigma_T)) \models Q$, which is done by scanning the universal semantics and so checking whether $h(\textit{chase}^B(D,\Sigma_T)) \models Q$. It remains to argue for the technique correctness.

\begin{theorem}
\label{th:decidability}
Given a set $\Sigma = \Sigma_T \cup \Sigma_E$ of warded TGDs and harmless EGDs, and a database $D$, if $D\cup\Sigma$ is satisfiable, for every query $Q$, it holds
$\textit{chase}^H(D,\Sigma)\models Q$ iff $\textit{chase}(D,\Sigma)\models Q$. 
\begin{proof}We prove the two directions of the implication separately.

\noindent($\Rightarrow$)  Let us proceed by contradiction and assume there exists a BCQ $Q: q\leftarrow\boldsymbol{\phi}(\mathbf{x})$ such that $\textit{chase}^H(D,\Sigma)\models Q$ and $\textit{chase}(D,\Sigma)\not\models Q$. Thus, there exists a homomorphism $\theta$ such that $\theta(\boldsymbol{\phi}(\mathbf{x}))\in \textit{chase}^H(D,\Sigma)$ and there does not exist any homomorphism $\theta^\prime$ such that $\theta^\prime(\boldsymbol{\phi}(\mathbf{x}))\in \textit{chase}(D,\Sigma)$. By harmlessness of $\Sigma_E$, we have that for a homomorphism $h$, it holds $h(\textit{chase}(D,\Sigma_T))\models Q$ iff $\textit{chase}(D,\Sigma)\models Q$. Therefore, there does not exist any fact (or conjunction thereof) $\boldsymbol{\varphi}(\mathbf{x}) \in \textit{chase}(D,\Sigma_T)$ such that $h(\boldsymbol{\varphi}(\mathbf{x})) \in \textit{chase}(D,\Sigma)$, with $\theta^\prime(\boldsymbol{\phi}(\mathbf{x}))=\boldsymbol{\varphi}(\mathbf{x})$. Then, since $\Sigma_T$ is warded, there does not exist any fact $\boldsymbol{\varphi}^\prime(\mathbf{x}) \in \textit{chase}^B(D,\Sigma^T)$, where $\boldsymbol{\varphi}^\prime(\mathbf{x}) = \boldsymbol{\varphi}(\mathbf{x})$ modulo the isomorphism of labelled nulls and such that $h(\boldsymbol{\varphi}^\prime(\mathbf{x}))\in\textit{chase}(D,\Sigma)$. Yet, we have that $\theta(\boldsymbol{\phi}(\mathbf{x}))\in\textit{chase}^H(D,\Sigma)$ and so, since $\Sigma_E$ is harmless and $h$ is surjective, there is some fact $\boldsymbol{\varphi}^\prime(\mathbf{x}) \in \textit{chase}^B(D,\Sigma_T)$ such that $\theta(\boldsymbol{\phi}(\mathbf{x}))=h(\boldsymbol{\varphi}^\prime(\mathbf{x}))$, which we assumed not to exist. Having reached a contradiction, we conclude $\textit{chase}(D,\Sigma)\models Q$.

\noindent($\Leftarrow$) Let us proceed again by contradiction and assume there exists a BCQ $q\leftarrow\boldsymbol{\phi}(\mathbf{x})$ such that $\textit{chase}(D,\Sigma)\models Q$ and $\textit{chase}^H(D,\Sigma)\not\models Q$. Thus, there exists a homomorphism $\theta$ such that $\theta(\boldsymbol{\phi}(\mathbf{x}))\in \textit{chase}(D,\Sigma)$ and there does not exist any homomorphism $\theta^\prime$ such that $\theta^\prime(\boldsymbol{\phi}(\mathbf{x}))\in \textit{chase}^H(D,\Sigma)$. By harmlessness of $\Sigma_E$, we have that for a homomorphism $h$, it holds $h(\textit{chase}(D,\Sigma_T))\models Q$ iff $\textit{chase}(D,\Sigma)\models Q$.
Then there exists a fact (or conjunction thereof) $\boldsymbol{\varphi}(\mathbf{x}) \in \textit{chase}(D,\Sigma_T)$ such that $h(\boldsymbol{\varphi}(\mathbf{x})) = \theta(\boldsymbol{\phi}(\mathbf{x}))\in \textit{chase}(D,\Sigma)$. Now, since $\Sigma_T$ is warded, we have that $\boldsymbol{\varphi}(\mathbf{x}) \in \textit{chase}^B(D,\Sigma)$, modulo fact isomorphism. However, since we have assumed that there does not exist a homomorphism $\theta^\prime$ such that $\theta^\prime(\boldsymbol{\phi}(\mathbf{x}))\in \textit{chase}^H(D,\Sigma)$, by harmlessness of $\Sigma_E$ and wardedness of $\Sigma_T$, there cannot exist any fact $\theta^\prime(\boldsymbol{\phi}(\mathbf{x}))=\boldsymbol{\varphi}(\mathbf{x}) \in \textit{chase}^B(D,\Sigma)$ such that $\boldsymbol{\varphi}(\mathbf{x}) \in \textit{chase}^H(D,\Sigma)$. Having reached a contradiction, we conclude $\textit{chase}^H(D,\Sigma)\models Q$.
\end{proof}
\end{theorem}

The technique for CQ answering we have discussed witnesses the decidability of our extension of Warded Datalog$^\pm$ with harmless EGDs. We have followed a pragmatic approach, where we basically augment the universal semantics of $D\cup \Sigma_T$ by chasing it under $\Sigma_E$: as we shall see, this idea lends itself to a practical application, once a space-efficient implementation for $\textit{chase}^B$ is provided. However, it does not allow us to immediately draw conclusions about the data complexity of the new fragment. In fact,
the \textsf{PTIME} membership proof of warded Datalog$^\pm$ hinges on the possibility to generate the ground semantics $\Sigma(D)_\downarrow = \{\overline{a} : \overline{a}\in\textit{chase}(D,\Sigma) \wedge \textit{dom}(a) = \mathbf{C} \}$ in polynomial time and does not consider the universal one. Hence, although our technique to chase the EGDs can be implemented in sub-polynomial time for a fixed set of EGDs~\cite{CGOP12}, nothing is said about the complexity of generating a suitable restriction of the universal semantics. 
On the other hand, it is to prove that the presence of harmless EGDs does not increase the complexity of computing the ground semantics. We follow this approach, revisiting and extending the results of~\cite{GoPi15}. The fundamental differences that makes our task more challenging is the need for a satisfiability check and a refocus of the \textsf{PTIME} membership proof of $\Sigma(D)_\downarrow$ in the presence of harmless EGDs. The following results first show time complexity with warded TGDs and harmless EGDs, under the assumption $D\cup\Sigma$ is satisfiable. Then, we illustrate how the satisfiability problem can be reduced to an instance of CQ answering  over warded TGDs and harmless EGDs, where $D\cup\Sigma$ is satisfiable.

\begin{theorem}
\label{th:complexity}
CQ answering for Warded Datalog$^\pm$ and harmless EGDs is \textsf{PTIME}-complete in data complexity, provided that the set of dependencies is satisfiable.
\begin{proof}[Proof (Sketch)]
The lower bound descends from Datalog being \textsf{PTIME}-hard~\cite{dantsin}. It is easy to verify that any set of Datalog dependencies are also a set of Warded Datalog$^\pm$ dependencies: in the absence of existential quantification, there are no affected variables and the warded condition is trivially respected. So we concentrate on the upper bound, considering a query $Q$, a database $D$ and a fixed set $\Sigma = \Sigma_T \cup \Sigma_E$ of warded TGDs and harmless EGDs. Our goal is proving that deciding whether $D\cup\Sigma \models Q$ is feasible in polynomial time in $D$.

Since we assume that $D\cup\Sigma$ is satisfiable, we can proceed with the following two steps. After a preliminary elimination of the negation (step~(1)), we just prove that checking whether an atom $\varPsi(\mathbf{t})$ belongs to the ground semantics $\Sigma^+(D^+)_\downarrow$ of the (negation-free) $\Sigma^+$ can be done in polynomial time (step~(2)). 
The correctness of the proof approach is straightforward.

\begin{itemize}[leftmargin=4mm]
\item By construction, we have that $D^+\cup\Sigma^+\models Q$ if for some $\varPsi(\mathbf{t}) \in \Sigma^+(D^+)_\downarrow$, it holds that $\varPsi(\mathbf{t})\in Q(\Sigma^+(D^+)_\downarrow)$. In other terms, to check whether $Q$ is satisfied, we need to scan the ground semantics $\Sigma^+(D^+)_\downarrow$. 
\item The size of $\Sigma^+(D^+)_\downarrow$ is polynomial in the size of $D$: the result has been proven for warded TGDs~\cite{ArenasGP18} and is not affected by harmless EGDs, which cannot in turn trigger TGDs and produce new facts.
\end{itemize}

\begin{enumerate}[leftmargin=4mm]

\item \underline{Elimination of Negation}: We build a database $D^+ \supseteq D$ and eliminate the negation from $\Sigma$ and $Q$, so that $D^+ \cup \Sigma^+ \models Q$ iff $D\cup\Sigma\models Q$. The procedure to build $D^+$ and $\Sigma^+$ is the standard one~\cite{GoPi15}, which applies to EGDs as well, and so $\Sigma^+ = \Sigma_T^+\cup\Sigma_E^+$. Since the negation in Warded Datalog$^\pm$ is stratified and ground, $\Sigma^+$ is computed from $\Sigma$ by iteratively replacing each negative atom $\neg a(\mathbf{x})$ with a positive atom $\overline{a}(\mathbf{x})$, where the extension of $\overline{a}(\mathbf{x})$ in $D^+$ is the complement of $a$ with respect to the ground semantics $\Sigma(D)_\downarrow$ (see step~(2)).

\item \underline{Building the Ground Semantics}: It remains to show the crucial technical lemma that the negation-free ---and we shall omit the plus superscript from hereinafter--- $\Sigma(D)_\downarrow$ can be built in polynomial time. As we have argued, this task has the same complexity as checking whether a fact $\varPsi(\mathbf{t})$ belongs to $\Sigma^+(D^+)_\downarrow$.

Let us first normalize $\Sigma_T$ so that every TGD $\sigma$ having more than one body atom is either \textit{head-ground}, i.e., all the head terms are constants, harmless and universally quantified variables, 
or \textit{semi-body-ground}, i.e., there exists at most one atom in $\textit{body}(\sigma)$ containing a harmful variable. We
partition $\Sigma_T$ into $\{ \Sigma_h, \Sigma_b \}$, where  $\Sigma_h$ is the set comprising all the head-ground TGDs, $\Sigma_b$ is the set of semi-body-ground TGDs. Without loss of generality, we will consider the presence of at most one existentially quantified variable in the TGDs.

\smallskip
We now provide a high-level description of a procedure
\textsc{EgdProofTree}($D,\Sigma_h,\Sigma_b,\Sigma_E,\varPsi(\mathbf{t})$) that takes as input a database $D$, a set of warded TGDs   partitioned into $\{ \Sigma_h, \Sigma_b \}$ as described, a set of harmless EGDs $\Sigma_E$, and a fact $\varPsi(\mathbf{t})$. The procedure accepts if $\varPsi(\mathbf{t}) \in \Sigma(D)_\downarrow$, rejects otherwise.
The procedure builds a proof tree.
Roughly speaking, a proof tree can be obtained from a proof (an execution of the chase, or a chase graph), by reversing the edges and unfolding the obtained graphs into a tree, by repeating some nodes. Conversely, we can obtain a chase graph, by reversing the edges of a proof tree and collapsing some nodes. Intuitively, building the two structures is equivalent.

We apply resolution steps starting from $\varPsi(\mathbf{t})$ until the database $D$ is reached. If not, $\varPsi(\mathbf{t})$ is rejected. This is done via an alternating procedure that builds the various branches of a proof tree for $\varPsi(\mathbf{t})$ in parallel universal computations, ensuring the compatibility of the branches. This is the main technical challenge, as we need to ensure that the introduced labelled nulls occurring in different branches of the proof tree represent the same term of $\Sigma(D)_\downarrow$, even when harmless EGDs are involved. 

\smallskip
\end{enumerate}
Let us introduce some conceptual tools we will use in the procedure. 
\begin{itemize}
\item the \textit{taint cause} $\chi(p[i])$ for a tainted position $p[i]$ is the set of EGDs of $\Sigma_E$ that can assign the value of a labelled null in $p[i]$. More formally, $\chi(p[i])$ can be inductively defined as the set of all the EGDs $\eta: \boldsymbol{\phi}(\mathbf{x}) \to x_i=x_j$ of $\Sigma_E$ such that: (i)~$p[i]$ is the position of $x_i$ (resp $x_j$) in $\textit{body}(\eta)$ and $p[i]$ is harmful in $\eta$, or, (ii)~for some TGD $\sigma\in\Sigma_T$ and some variable $v\in\textit{body}(\sigma)\cap\textit{head}(\sigma)$, we have that $v$ appears in a body (resp.\ head) position that is in $\chi(p[i])$ and in position $p[i]$ in $\textit{head}(\sigma)$ (resp.\ $\textit{body}(\sigma)$). For the sake of simplicity, in the rest of the proof we will assume that $x_i$ is the harmful variable in the EGD body.

\item Given an atom $\varPsi$, a TGD $\sigma\in \Sigma_h$, and a homomorphism $\theta$ from $\textit{head}(\sigma)$ to $\varPsi$, a \textit{head-ground TGD resolution step} $\varPsi\triangleright^h_{\theta_\sigma}$ produces the set of atoms $\{ \theta_\sigma^B(b) : b \in \textit{body}(\sigma) \}$, where $\theta_\sigma^B$ is the extension of $\theta$ that assigns variables appearing $\textit{body}(\sigma) \setminus \textit{head}(\sigma)$ to fresh labelled nulls or constants from $\textit{dom}(D)$.

\item Given an atom $\varPsi$, a TGD $\sigma\in \Sigma_h$, and a homomorphism $\theta$ from $\textit{head}(\sigma)$ to $\varPsi$, a \textit{semi-body-ground TGD resolution step} $\varPsi\triangleright^h_{\theta_\sigma}$ produces the set of atoms $\{ \theta_\sigma^B(b) : b \in \textit{body}(\sigma) \}$, where $\theta_\sigma^B$ is the extension of $\theta$ that assigns variables appearing $\textit{body}(\sigma) \setminus \textit{head}(\sigma)$ 
in other atoms than the ward, to constants from $\textit{dom}(D)$.
The variables in the ward are mapped to either labelled nulls or constants from $\textit{dom}(D)$.

\item Given an atom $\varPsi(\mathbf{t})$ where $\pi=\varPsi(t[i])$ is a tainted position, an EGD $\eta: \boldsymbol{\phi}(\mathbf{x}) \to x_i=x_j$, with $\eta \in \chi(\pi)$, a homomorphism $\theta$ that maps $x_j$ to $\varPsi(t[i])$, a \textit{an EGD resolution step} $\varPsi(\mathbf{t})\triangleright^\pi_{\theta_\eta}$ produces the set of atoms $\{ \theta_\eta^B(b) : b \in \textit{body}(\eta) \}\cup\{\varPsi^\prime(\mathbf{t}\} )\}$, 
where $\theta_\eta^B$ is the extension of $\theta$ that 
assigns variables appearing $\textit{body}(\sigma)$
to either fresh labelled nulls or constants from $\textit{dom}(D)$,
and $\varPsi^\prime(\mathbf{t})$ is obtained from $\varPsi(\mathbf{t})$, by replacing the value at $\varPsi(t[i])$ with $\theta_\eta^B(x_i)$.
The intuition behind a partial EGD resolution step is explaining the value $\varPsi(t[i])$ of one single tainted position via one EGD application.

\item Given a set of atoms $S$, and a set of labelled nulls $N \subseteq \mathbf{N}$, a partition $\Pi_N(S)$ is known as as an \textit{[N]-linking} partition~\cite{ArGP14} if $\Pi_N(S)$ is a partition for $S$ and for each labelled null $\nu \in (\textit{dom}(S)\cap\mathbf{N})\setminus N$ there is exactly one subset $S_i\in \Pi_N(S)$ in which $\nu$ appears. A partition $\Pi_N(S)$ is defined as \textit{[N]-optimal} if it is [N]-linking and there does not exist any other [N]-linking partition $\Pi^\prime_N(S)$ such that $|\Pi^\prime_N(S)|>|\Pi_N(S)|$. For example, consider the set $S = \{p(c, z_1),p(z_1, z_2),p(z_2, z_3),p(z_3, z_4)\}$, where $z_1,z_2,z_3,z_4\in\mathbf{N}$ and let $N=\{z_2,z_3,z_4\}$. 
The partition $\{\{p(c, z_1),p(z_1, z_2)\}, \{p(z_2, z_3),p(z_3, z_4)\}\}$ is 
[N]-linking as every $z_i \in \textit{dom}(S)\cap\mathbf{N}$ occurs exactly in one component. Yet, it is not [N]-optimal, because the partition $\{\{p(c, z_1),p(z_1, z_2)\}, \{p(z_2, z_3)\}, \{p(z_3, z_4)\}\}$ is still [N]-linking, but has higher cardinality. It turns out, the latter is [N]-optimal, as $\{p(c, z_1),p(z_1, z_2)\}$ cannot be split without separating the $z_1$ occurrences.

\end{itemize}

\smallskip\noindent
The procedure performs the following steps. Let us first initialize $S=\{\varPsi(\mathbf{t})\}$. Start from step~(a).
\smallskip

\begin{enumerate}
    \item \textbf{EGD resolution}. For each $\varPsi(\mathbf{t})\in S$, guess whether skipping to step~(b). If not, if $\varPsi(\mathbf{t})$ does not contain any tainted position, go to step~(b). If $\varPsi(\mathbf{t})$ contains a tainted position $\pi$, guess an EGD $\eta \subseteq \Sigma_E$ such that $\eta\in\chi(\pi)$ and let  $\theta_{\eta}$ be a homomorphism from $x_j$ to $\varPsi(\mathbf{t[\pi]})$.
    Apply the resolution step $\varPsi(\mathbf{t})\triangleright^{\pi{}}_{\theta_{\eta}}\{b_{\varPsi_1},\ldots,b_{\varPsi_n}\}$  and let $S_\varPsi=\{b_{\varPsi_1},\ldots,b_{\varPsi_n}\}$ be the generated atoms. Then, with $S^+=\bigcup_{\varPsi\in S}S_{\varPsi}$, we compute $\Pi_N(S^+)=\{S_1^+,\ldots,S_n^+\}$ as the [$\emptyset$]-optimal partition of $S^+$. Universally select and prove in parallel computations each element of $\Pi_{\emptyset}(S^+)$, recursively from step~(a). The goal of this partitioning strategy is keeping in the same universal computation all the labelled nulls appearing in more than one branch.
    \item \textbf{TGD $\in \Sigma_h$ resolution}. If $\varPsi(\mathbf{t}) \in D$ then accept. If $\textit{dom}(\varPsi(\mathbf{t}))\cap\mathbf{N}\neq\emptyset$, then go to step~(c) with $\{\varPsi(\mathbf{t})\}$. Else, guess a TGD $\sigma \in \Sigma_h$ and a homomorphism $\theta_\sigma$ from $\textit{head}(\sigma)$ to $\varPsi(\mathbf{t})$. Apply the resolution step  $\varPsi(\mathbf{t})\triangleright^h_{\theta_\sigma}S$, where $S=\{b_1,\ldots,b_n\}$ are the generated atoms. If such $\sigma$ or $\theta_\sigma$ cannot be found, then reject. Let $\Pi_{\emptyset}(S)$ be an [$\emptyset$]-optimal partition of $S$. Universally select each $S\in \Pi_{\emptyset}(S)$ and proceed as follows: if 
    $S=\{\varPsi^\prime(\mathbf{t}^\prime)\}$ and $\textit{dom}(S_i) \subseteq \textit{dom}(D)$, recursively call step~(b) for  $\varPsi^\prime(\mathbf{t}^\prime)$; else, that is, if $\textit{dom}(S)\cap\mathbf{N}\neq\emptyset$, go to step~(c).
    
    \item \textbf{TGD $\in \Sigma_b$ resolution}. Let us initially mark all the labelled nulls occurring in $S$,  by storing them into a set $R_{S} = \{(\nu,\epsilon) : \nu \in \textit{dom}(S)\cap\mathbf{N}\}$.
    For each $a\in S$ guess a rule $\sigma\in\Sigma_b$ and a homomorphism $\theta_\sigma$ from $\textit{head}(\sigma)$ to $\varPsi(\mathbf{t})$. If, for any $a$, the labelled null $\overline{\nu}$ occurs in an existentially quantified position of $\sigma$, then mark it by updating the corresponding element in $R_s$ as $(\overline{\nu},a)$.
    Apply the resolution step  $\varPsi(\mathbf{t})\triangleright^b_{\theta_\sigma}S_a$, where $S_a=\{b_{a_1},\ldots,{b_{a_m}}\}$ are the generated atoms. If such $\sigma$ or $\theta_\sigma$ cannot be found, then reject.
    Then, with $S^+=\bigcup_{a\in S}S_a$, we let $N = \{\nu \in \textit{dom}(S^+)\cap\mathbf{N} : (\nu,x) \in R_s, x\neq \epsilon\}$ and compute $\Pi_N(S^+)=\{S_1^+,\ldots,S_n^+\}$ as the [$N$]-optimal partition of $S^+$. Intuitively, we are separating labelled nulls for which a generating atom, via its existential quantifier, has been found, while we are keeping the others linked. In this way, it is sufficient to store the generating atom along with the invented labelled nulls, to be able to safely merge different branches. Then for each $S_i\in\Pi_N(S^+)$,
     the new labelled nulls in $S_i^+$ are marked, we universally select $S\in\{S_1^+,\ldots,S_n^+\}$ and proceed as follows: if 
    $S=\{\varPsi^\prime(\mathbf{t}^\prime)\}$ and $\textit{dom}(S_i) \subseteq \textit{dom}(D)$, recursively call step~(b) for  $\varPsi^\prime(\mathbf{t}^\prime)$; else, that is, if $\textit{dom}(S)\cap\mathbf{N}\neq\emptyset$, go to step~(c).
\end{enumerate}

\smallskip\noindent The correctness of the procedure follows by construction and harmlessness: if existing, an accepting computation for $\varPsi(\mathbf{t})$ can be found by first resolving all the EGDs as they do not affect the applicability of TGDs and then, once step~(a) non-deterministically skips to~(b), by resolving the TGDs (by applying the procedure from~\cite{ArGP14}).

\smallskip
Let us briefly discuss the labelled null life cycle throughout the procedure. A fresh labelled null can be invented in resolution steps in different cases. An EGD step can introduce at most one labelled null for each variable in the body other than $x_j$; it can also introduce one more labelled null as a placeholder in the tainted position of $\varPsi^\prime$. A TGD step introduces nulls whenever an atom is resolved with a rule having body variables not appearing in the head. Nulls then propagate with resolution steps, copied from the rule heads into the resolving atoms. A null is suppressed when it is resolved via an existentially quantified variable in a rule head. Clearly, the overall assignment of the nulls must be coherent in all the universal computations. Therefore it is essential that the atoms where the null $\nu$ is invented in different parallel computations are isomorphic to $\theta(\textit{head}(\sigma))$. To ensure that, the partitioning technique keeps together within the partition that contains $\nu$, the atom $\theta(\textit{head}(\sigma))$. In this way, the different branches can be eventually merged in a sound way, with a global coherence of the nulls.

\smallskip\noindent\underline{\textit{Space Complexity of the Steps}}. We encode \textsc{EgdProofTree}($D,\Sigma_h,\Sigma_b,\Sigma_E,\varPsi(\mathbf{t})$) as 
$\textsc{EGDOnlyProofTree}(D,\Sigma_E,\varPsi(\mathbf{t}))$\\ $\bigcup_i \textsc{TGDProofTree}(\mathbf{a_i},\Sigma_h,\Sigma_b,\varPsi(\mathbf{t}))$, that is, a partial proof tree built by applying only steps~(a) followed by multiple proof trees built with the application of steps (b) and (c) on all the leaves $\mathbf{a_i}$ of the first proof tree. By harmlessness of $\Sigma_E$, this decomposition is correct as when a computation of (a) non-deterministically skips to (b), no application of step~(a) will be performed again. Intuitively, once the EGDs have been used to explain $\varPsi(\mathbf{t})$ by a set of TGD outputs, only TGDs are used in the following resolution steps, coherently with harmlessness. It has been proven that a single $\textsc{TGDProofTree}(D,\Sigma_h,\sigma_b,\mathbf{a_i})$ can be generated in polynomial time, by showing that steps (b) and (c) use $O(\textit{log}(\textit{dom}(D)))$ space at each step and remembering that that alternating logarithmic space (\textsf{ALOGSPACE}) coincides with \textsf{PTIME}~\cite{ArGP14}. It remains to show that the application of steps (a) requires polynomial time.

\smallskip We start by proving that the size of each component of a $[\emptyset]$-optimal partition of $S^+$ in $\textsc{EGDOnlyProofTree}$ 
is at most $AB(B+1)$, where $B$ is $\textit{max}_{\eta\in\Sigma_E}|\textit{body}(\eta)|$ and $A$ is the maximum arity of an atom of $\textit{body}(\eta)$. 
Intuitively, each step adds at most $B+1$ new atoms, until all the possibly tainted variables ($A$) of all the body atoms ($B$) have been resolved. At that point, each component of $\Pi_\emptyset(S_i^+)$ stably contains at most $AB(B+1)$ elements by construction. Each subsequent steps removes and re-adds at most $B+1$ atoms.

More formally, let us proceed by induction on the number of applied steps~(a). As the base case, we have that the first partitioning step produces $\Pi_\emptyset(S_0^+) = \{S^0_{\varPsi_1},\ldots,S^0_{\varPsi_n}\}$. Each $S^0_{\varPsi_i} \in \Pi_\emptyset(S^+)$ contains at most $B+1$ elements, as in the worst case, we resolve a (ground) tainted position of $\varPsi({\mathbf{t}})$ via an EGD resolution step that produces the maximum number of atoms of $\textit{body}(\eta)$, all sharing a fresh labelled null, plus one atom $\varPsi^\prime({\mathbf{t}})$, sharing the labelled null $\theta_\eta^B(x_i)$ used to replace the tainted position with one atom in $\textit{body}(\eta)$. In the inductive case, we consider a component $S^k_{\varPsi_i} \in \Pi_{\emptyset}(S^+_k)$ generated during the $k$-th step, with $k>1$ by resolving the atoms $\mathbf{a}$ in $S^{k-1}_{\varPsi_j} \in \Pi_{\emptyset}(S^+_{k-1})$. By inductive hypothesis we assume $|S^{k-1}_{\varPsi_j}| \le AB(B+1)$. 
Two cases are possible for each resolution step 
$\varPsi(\mathbf{t})\triangleright^{\pi{}}_{\theta_{\eta}}\{b_{\varPsi_1},\ldots,b_{\varPsi_n}\}$ applied on an atom $\mathbf{a^*}$ before the partitioning step $k$: (i) the variable $\mathbf{a}^*[\pi]$ is ground; (ii) $\mathbf{a}^*[\pi]$ is a labelled null $\nu$ introduced in a previous resolution step by applying $\theta_\eta^B(x_i)$, and $\pi$ is a tainted position for $\textit{body}(\eta)$. In case (i), as $\pi$ is ground, no resolution step has ever been applied on that tainted position; hence, at least one EGD $\eta$ has never been applied in a resolution step. By applying the inductive hypothesis it follows that $S^{k-1}_{\varPsi_j}$ contains at most $(B+1)(AB-1)$ atoms, and $|S^{k}_{\varPsi_i}| \le AB(B+1)$. In case~(ii), the resolution step generates at most $B+1$ new atoms in the same partition, including the atom $\mathbf{a^*}^\prime$ obtained by replacing the tainted variable $x_i$ of $\mathbf{a^*}$ in position $\pi$, with $\theta_\eta^B(x_i)$. The introduction of a fresh labelled null in $\mathbf{a^*}^\prime[\pi]$, disconnects this atom from at least $B+1$ other atoms potentially in the same partition, whence $|S^{k}_{\varPsi_i}| \le |S^{k-1}_{\varPsi_j}| \le AB(B+1)$. 

\smallskip
Having set a bound on the size of the $[\emptyset]$-optimal components of $S^+$, it is easy to see that we need to remember at most $(AB(B+1))^2$ atoms, at each step. The space needed to represent such atoms is polynomial in $\Sigma_E$ and logarithmic in $|\textit{dom}(D)|$. If $\Sigma_E$ is fixed, it follows that $\textsc{EGDOnlyProofTree}(D,\Sigma_E,\varPsi(\mathbf{t}))$ uses $O(\textit{log}|\textit{dom}(D)|)$ space at each step of its computation, and thus the proof tree can be generated in $\textsf{PTIME}$. In total, remembering that \textsf{ALOGSPACE} coincides with \textsf{PTIME}, we conclude that $\textsc{EGDOnlyProofTree}$ and so $\textsc{EGDProofTree}$ can be generated in polynomial time.
\end{proof}
\end{theorem}
We now show how the satisfiability problem can be reduced to CQ answering over warded TGDs and harmless EGDs where $D\cup\Sigma$ is satisfiable and provide the full complexity result.

\begin{theorem}
\label{th:complexityext}
CQ answering for Warded Datalog$^\pm$ and harmless EGDs is \textsf{PTIME}-complete in data complexity.
\begin{proof}
Let $Q$ be a query, $D$ be a database, and $\Sigma=\Sigma_T\cup\Sigma_E$ a fixed set of warded TGDs and harmless EGDs. We proceed by case distinction based on whether $D\cup\Sigma$ is satisfiable. If $D\cup\Sigma$ is satisfiable, by Theorem~\ref{th:complexity}, we can conclude that the problem of deciding $D\cup\Sigma\models Q$ is \textsf{PTIME} in data complexity. It remains to show that deciding on the satisfiability of $D\cup\Sigma$ can be done in $\textsf{PTIME}$.

\smallskip
We first give a sketch of the proof, and then describe the construction. Observe that the only way $D\cup\Sigma$ is not satisfiable is a hard violation of an EGD, i.e., equating two constants.
Intuitively, the main idea behind this proof is to separate the effects of EGDs into (i) hard violations, i.e., equating constants, and (ii) soft violations, i.e., equating labelled nulls with other values (nulls or constants).
To detect and correct soft violations, we rewrite each EGD $\eta$ into a pair of EGDs $\eta^\prime$ and $\eta^{\prime\prime}$, where we alternatively force $x_i$ and $x_j$ to bind to non-ground values. Hard violations, and thus satisfiability, are then detected iff none of the queries $Q_V$ of a set $\mathbf{Q}_V$ of check queries for hard violations hold. The encoding is done as follows.

We first construct sets $D'$ and $\Sigma'_E$, which will make up $D^\prime\cup\{\Sigma_T\cup\Sigma_E^\prime\}$ against which we will perform our queries.
We extend $D$ into $D^\prime$ by adding facts $\textit{neq}(c_1,c_2)$ for each pair of distinct constants $c_1,c_2 \in \textit{dom}(D)$. We define a new set of EGDs $\Sigma_E^\prime$, containing two EGDs, $\eta^\prime$ and $\eta^{\prime\prime}$, for each $\eta \in \Sigma_E$, built as follows. For each EGD $\eta: \boldsymbol{\phi}(\mathbf{x}) \to x_i=x_j$ of $\Sigma_E$, we let: $\eta^\prime: \boldsymbol{\phi}(\mathbf{x}),\textit{null}(x_i) \to x_i=x_j$, and
$\eta^{\prime\prime}: \boldsymbol{\phi}(\mathbf{x}),\textit{null}(x_j) \to x_i=x_j$. 
The artificial $\textit{null}$ predicate forces $x_i$ (resp.\ $x_j$) to bind only to labelled nulls. It is easy to check that assuming $D\cup\Sigma$ is satisfiable, $D\cup\Sigma_T$ is logically equivalent to $D\cup\{\Sigma_T\cup\Sigma_E^\prime\}$, as the only failing cases arise when both $x_i$ and $x_j$ are bound to constants in some EGD and $x_i\neq x_j$ (i.e., hard violations). On the other hand, such EGD activations produce no effects if $x_i=x_j$.

We now construct the set of queries $\mathbf{Q}_V$ that check for hard violations.
For each $\eta \in \Sigma_E$, we add to $\mathbf{Q}_V$ the following BCQ: 
$q\leftarrow \boldsymbol{\phi}(\mathbf{x}), \textit{neq}(x_i,x_j)$. We have that $D^\prime\cup\{\Sigma_T\cup\Sigma_E^\prime\}$ is always satisfiable, since there cannot be hard violations by construction. Moreover, as $\Sigma_E$ is harmless, $\Sigma_E^\prime$ is harmless as well. Thus, by Theorem~\ref{th:complexity}, $D^\prime\cup\{\Sigma_T\cup\Sigma_E^\prime\}\models Q_V$ can be checked in $\textsf{PTIME}$, for each $Q_V\in \mathbf{Q}_V$.
\end{proof}
\end{theorem}

\subsection{Harmless EGDs and Warded Semantics}

Towards a full characterization of the relationship among warded semantics and harmless EGDs in CQ answering, let us consider the following set of warded TGDs.

\begin{example}
\label{ex:hje_and_egds}
\begin{align*}
  \textit{p}(x,k), p(y,k), x \neq y \rightarrow \exists z~c(z,x,y) \tag{$\sigma_1$} \\
  \textit{c}(z,x,y) \rightarrow s(z,x) \tag{$\sigma_2$} \\
  \textit{c}(z,x,y) \rightarrow s(z,y) \tag{$\sigma_3$} \\
\end{align*}
\textit{This is, again, a clustering scenario where elements $x$ and $y$ having a common feature $k$ are assumed to belong to the same cluster $z$ ($\sigma_1$). Then, $s$ contains the clusters each element belongs to ($\sigma_2, \sigma_3$). In $\sigma_1$, we pose $x\neq y$ for the sake of simplicity. $\blacksquare$}
\end{example}

Consider rules in Example~\ref{ex:hje_and_egds} with respect to database $D=\{p(1,2), p(2,2)\}$ and the atomic query $Q_1: q(z,x) \leftarrow s(z,x)$. With the warded semantics, we can have one of the following results: $\{(z=z_1,x=1), (z=z_2,x=2)\}$, $\{(z=z_1,x=1), (z=z_1,x=2)\}$, $\{(z=z_2,x=1), (z=z_1,x=2)\}$, $\{(z=z_2,x=1), (z=z_2,x=2)\}$. The specific binding depends on the particular sequence of chase steps that is applied. While all the results are equivalent under warded semantics, they say nothing about whether elements $1$ are $2$ are in the same cluster, because labelled nulls lose their identity.
Were we interested in knowing exactly whether two elements belong to the same cluster, we would resort to the CQ $Q_2: q(x,y) \leftarrow s(z,x),s(z,y)$. The CQ answer would be $\{(x=1,y=2),(x=2,y=2)\}$.

\smallskip
So, what is the impact of a harmless EGD 
(e.g., $\eta = s(z,x),s(z^{\prime},x) \to z = z^{\prime}$) unifying the cluster identifiers
in our example?
This question can be put in a more general perspective and formulated in multiple equivalent ways: are harmless EGDs compatible with the warded semantics? In a CQ with warded rules, can we control the identity of output labelled nulls via harmless EGDs? Can EGDs themselves rely on the equality of labelled nulls under warded semantics?

It turns our that the answer to these questions is negative. In Warded Datalog$^\pm$, the irrelevance of null identity is exploited to eliminate joins on labelled nulls. In the warded semantics, facts are considered equivalent modulo isomorphism. On the other hand, EGDs assign specific values to labelled nulls, invalidating this approach.

\smallskip
Consider again Example~\ref{ex:hje_and_egds} and the database $D = \{p(1,A),p(2,A),p(3,A)\}$. Figure~\ref{fig:chase} shows a portion of $\textit{chase}(D,\Sigma)$. As we have seen, we cannot rely on warded semantics and need to generate multiple isomorphic copies of the same facts, e.g., for $s(z,1)$, $s(z,2)$, $s(z,3)$. In fact, it is only from the comparison of different isomorphic copies that EGDs can enforce $z_1 = z_2 = z_3$, whilst according to warded semantics only one fact $s$ for each different value for $q[2]$ would be produced. It is however clear that not all the isomorphic copies are needed: for example, once from $s(z_1,1)$ and $s(z_3,1)$ we establish $z_1=z_3$ and from $s(z_1,2)$ and $s(z_2,2)$ we establish that $z_1=z_2$, there is no need to duplicate $s(z_2,3)$, as $z_3=z_2$ and $z_3 = z_1$ hold by transitivity. 

\begin{figure}

\centering
\includegraphics[scale = 0.20]{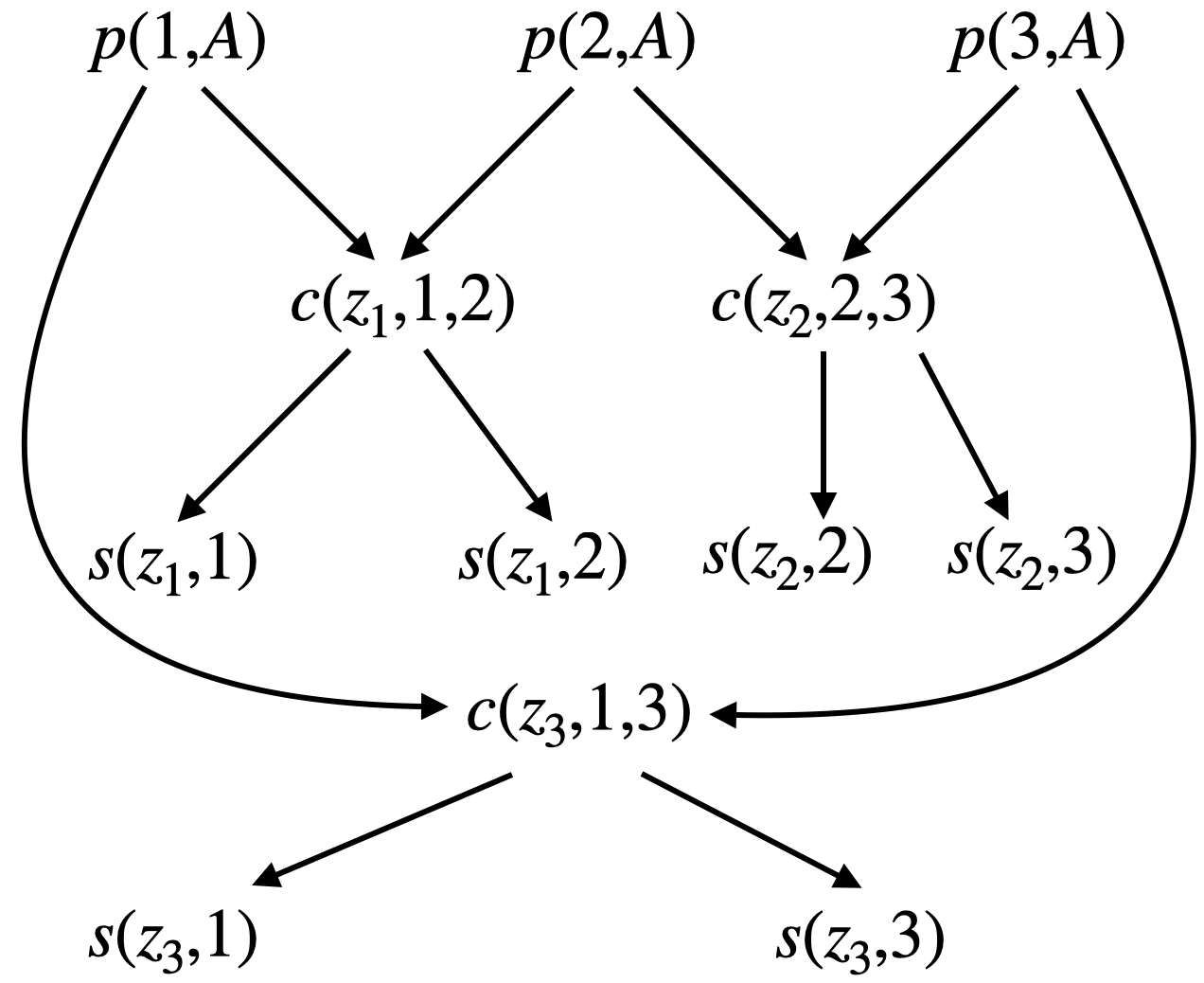}
\centering
\caption{Chase application for Example~\ref{ex:hje_and_egds}.}
\label{fig:chase}
\end{figure}

\smallskip
In Theorem~\ref{th:decidability}, we leveraged the universal semantics to build a finite and CQ-equivalent $\textit{chase}^H(D,\Sigma) \subseteq \textit{chase}(D,\Sigma)$, where EGDs are applied to fixpoint. As we cannot choose $\textit{chase}^W$ to provide an efficient implementation of $\textit{chase}^B$ (used in $\textit{chase}^H)$, it is our goal in this section to define a \textit{relaxed warded semantics}  $\textit{chase}^{B_W}(D,\Sigma) \subseteq \textit{chase}(D,\Sigma)$, that is a finite restriction of $\textit{chase}(D,\Sigma)$, CQ-equivalent to it even in the presence of harmless EGDs.

Figure~\ref{fig:semantics} summarizes the described setting, highlighting the different EGD classes, whereas containment relationships between the semantics (and the respective chase variants) are in Figure~\ref{fig:chases}.

\begin{figure}[h]
\begin{subfigure}{.5\textwidth}
\centering
\includegraphics[scale = 0.30]{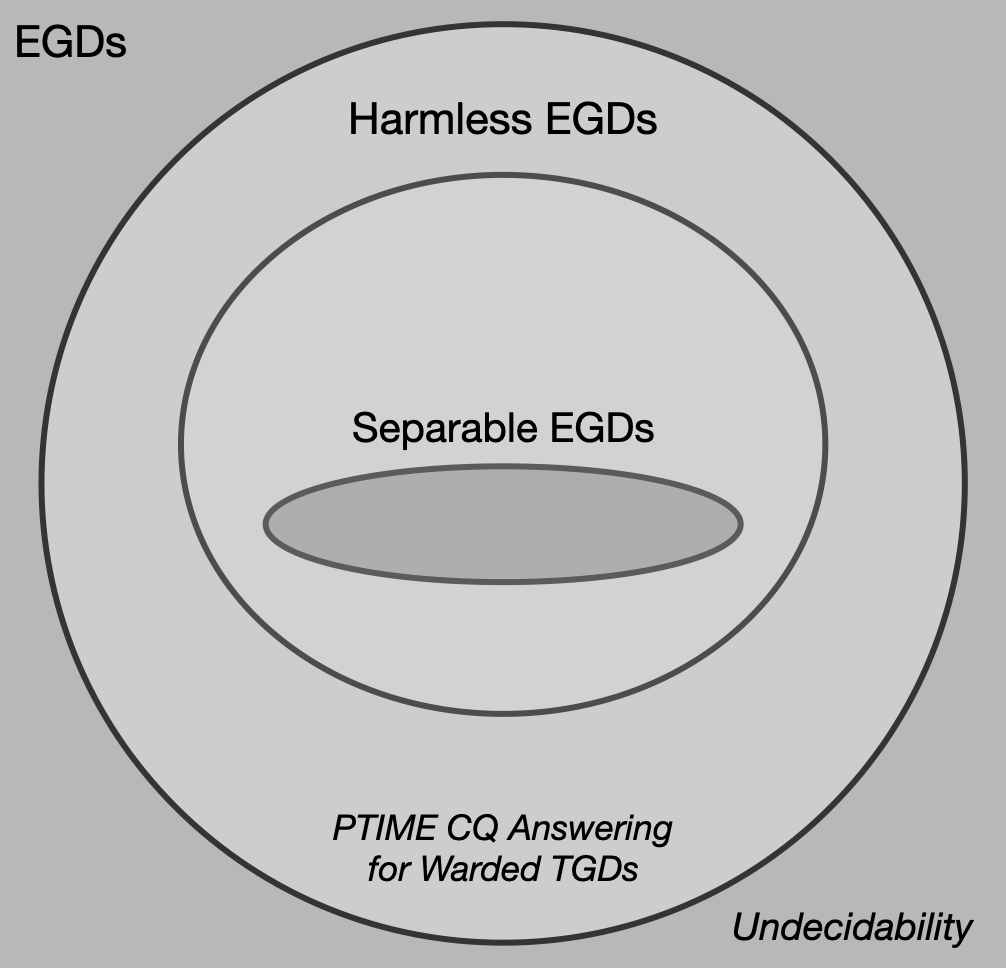}
\caption{Syntactic containment of EGD classes.}
\label{fig:semantics}
\end{subfigure}%
\begin{subfigure}{.5\textwidth}
\centering
\includegraphics[scale = 0.22]{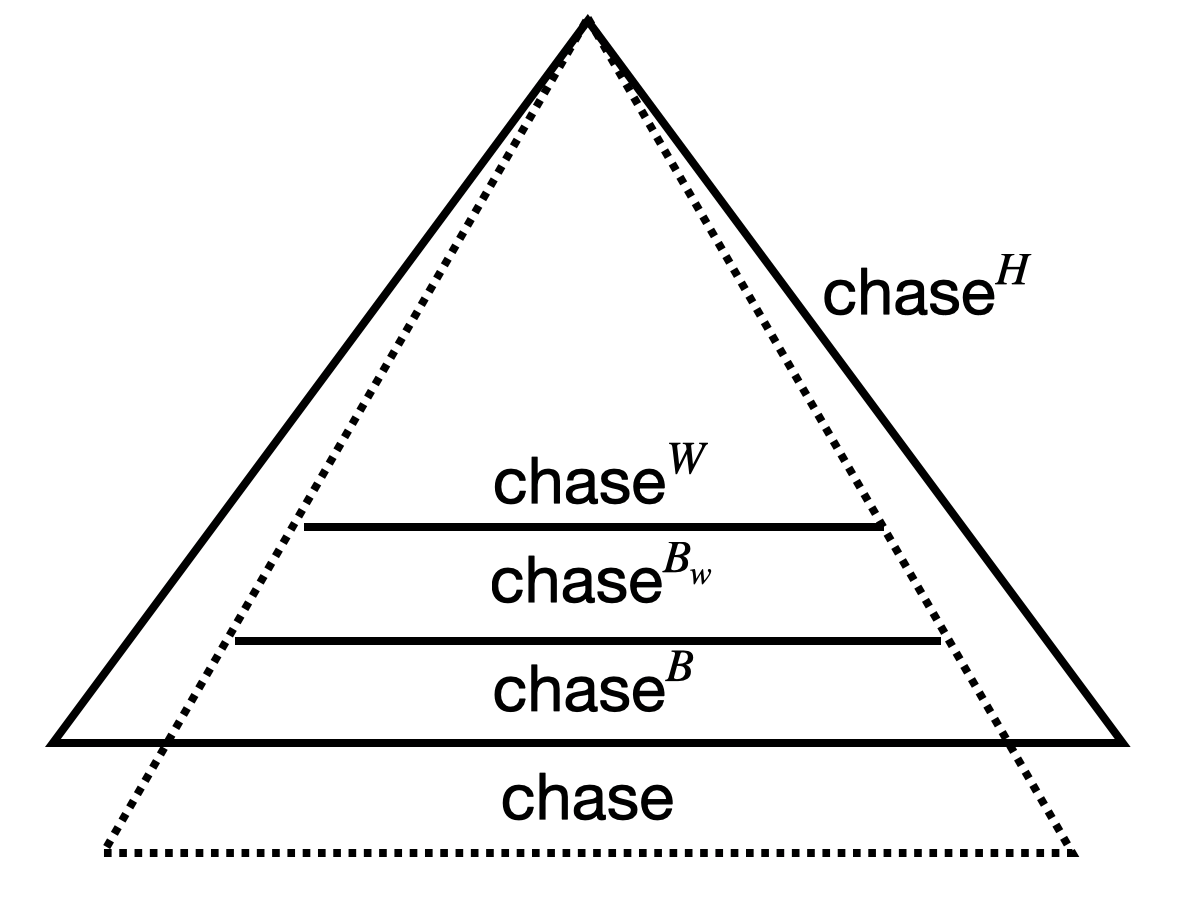}
\caption{Containment of the involved semantics: (1)~$\textit{chase}^W$: warded semantics; (2)~$\textit{chase}^{B_W}$: relaxed warded semantics, an extension of $\textit{chase}^W$, where multiple isomorphic copies are allowed to enable the application of EGDs; (3)~$\textit{chase}^{B}$: universal semantics of Warded Datalog$^\pm$; (4)~$\textit{chase}^H$ is a finite portion of the full chase, CQ-equivalent to the full infinite chase in the presence of harmless EGDs; (5)~full potentially infinite chase.}
\label{fig:chases}
\end{subfigure}
\label{fig:myfig}
\caption{~}
\end{figure}

\medskip
Towards a characterization of $\textit{chase}^{B_W}$, we need to introduce some more of our theoretical tools.

\medskip
The \textit{chase graph} $\mathcal{G}(\Sigma,D)$ for a database $D$ and a set of dependencies $\Sigma$ is the directed graph consisting of $\textit{chase}(D,\Sigma)$ as the set of nodes and having an edge from \textbf{a} to \textbf{b} iff \textbf{b} is obtained from \textbf{a} and possibly other facts by the application of a TGD of $\Sigma$. The \textit{warded forest} $\mathcal{W}(\mathcal{G})$ of a chase graph is the the subgraph that consists of all nodes of the chase graphs, all edges of the chase graph that correspond to the application of linear rules (i.e., having one single body atom), and one edge for each nonlinear rule ---namely the one from the fact bound to the ward~\cite{BeSG18}. The connected components of a warded forest are determined by the joins that involve constants, in fact, each component contains only edges representing linear rules or joins involving dangerous variables. The $\textit{subtree}(\mathcal{W},\textbf{v})$ is the subtree of $\mathcal{W}(\mathcal{G}(\Sigma,D))$ induced by all the descendants of $\textbf{v}$.
The $\textit{subgraph}(\mathcal{G},\textbf{v})$ is the subgraph of $\mathcal{G}(D,\Sigma)$ induced by all the descendants of $\textbf{v}$.

\begin{definition}
\label{def:track}
\textit{Let $\Sigma$ be a set of warded TGDs, $D$ a database and $\textbf{a}$ a fact of $\textit{chase}(D,\Sigma)$. Let $d(x,y)$ be the graph distance between two facts $x$ and $y$ of $\mathcal{G}(\Sigma,D)$ (we assume $d(x,x)=0$). We define $\mathbf{b} = \textbf{track}(\textbf{a})$ as the $d(\mathbf{b},\mathbf{a})$-maximal fact $\mathbf{b}$ of $\mathcal{G}$ such that $\mathbf{a} \in \textit{subtree}(\mathcal{W},\mathbf{b})$. Intuitively, $\mathbf{b}$ is the ultimate root of the tree to which $\mathbf{a}$ belongs in $\mathcal{W}(\mathcal{G})$.}
\end{definition}

\begin{definition}
\label{def:track-isomorphism}
Let $\Sigma$ be a set of warded TGDs, $D$ a database, and $T$ a fact of the chase graph $\mathcal{G}(\Sigma,D)$. We say that 
two facts $\textbf{a}$ and $\textbf{b}$ are \emph{$T$-isomorphic},
if they are isomorphic and have the same track $T=\textbf{track}(\textbf{a}) = \textbf{track}(\textbf{b})$.
\end{definition}

Definition~\ref{def:track-isomorphism} restricts the general notion of isomorphism between facts and applies it only when they actually originate from the same root of the warded forest. Intuitively, for two isomorphic facts deriving from different tracks, we do not make any assumption on the equality of the involved labelled nulls. Conversely, for two isomorphic facts having the same track, the value of their labelled nulls is actually encoding the same information. 
For example, in Figure~\ref{fig:chase}, facts $s(z_1,1)$ and $s(z_3,1)$  are isomorphic but not $T$-isomorphic, as: $\textbf{track}(s(z_1,1)) = c(z_1,1,2)$ and $\textbf{track}(s(z_3,1)) = c(z_3,1,3)$.

\begin{theorem}
\label{th:t-isomoprhism_implies_subtree_isomorphism}
Let $\Sigma$ be a set of warded TGDs, $D$ be a database. If two facts \textbf{a} and \textbf{b} of $\mathcal{G}(\Sigma,D)$ are T-isomorphic, then $\textit{subgraph}(\mathcal{W}(\mathcal{G}),\textbf{a})$ and $\textit{subgraph}(\mathcal{W}(\mathcal{G},\textbf{b}))$ are isomorphic, according to the standard notion of graph isomorphism.
\begin{proof}
It can be easily proven as a special case of Theorem~2 in~\cite{BeSG18} which holds in the more general case of \textbf{a} and \textbf{b} being isomorphic, given that by Definition~\ref{def:track-isomorphism}, T-isomorphism necessitates isomorphism.
\end{proof}
\end{theorem}

Standard isomorphism check already suggests an effective algorithm to build $\textit{chase}^W(D,\Sigma_T)$, which guarantees correctness of CQ answering. It is used in the {\sc Vadalog} system and, roughly, consists in pruning the exploration of a chase graph branch, whenever a fact is met that is isomorphic to an already generated one. 
Theorem~\ref{th:t-isomoprhism_implies_subtree_isomorphism} puts a more stringent notion of isomorphism into action in sets of warded rules. 
We claim that the same chase pruning algorithm can be applied on the basis of $T$-isomorphism: as it is a stricter notion, it will give rise to a new chase variant $\textit{chase}^{B_W}(D,\Sigma) \supseteq \textit{chase}^{W}(D,\Sigma)$, which will prove to preserve enough information about labelled nulls to enable CQ answering with harmless EGDs. Clearly, $T$-isomorphism is an equivalence relation and with the next definition, we see how it can be used to declaratively describe the structure of this new chase variant.

\begin{definition}
\label{def:chaseBW}
\textit{Given a database $D$, a set of warded rules $\Sigma$, let $\mathbf{Q}$ be the quotient set $\textit{chase}(D,\Sigma)/\mathcal{T}$, induced by the $T$-isomorphism relation $\mathcal{T}$ on the standard chase. We define the relaxed warded semantics $\textit{chase}^{B_W}(D,\Sigma)$ as the set of all the class representatives of $\mathbf{Q}$, one for each equivalence class of $\mathbf{Q}$.}
\end{definition}

Definition~\ref{def:chaseBW} allows to partition the infinite chase  $\textit{chase}(D,\Sigma)$ into a finite number of equivalence classes, so that CQ answering in the presence of harmless EGDs can be performed on them. First, we need to discuss the boundedness of $\textit{chase}^{B_W}(D,\Sigma)$ and then to argue for the correctness of CQ answering with EGDs.

\begin{theorem}
\label{th:chaseBW-boundedness} Let $\mathbf{S}$ be a database schema and $w$ the maximal arity of its predicates. Given a database $D$ and a set of warded TGDs $\Sigma$, both defined for $\mathbf{S}$, let $P$ be the set of pairs $\langle T, \textbf{a} \rangle$ where $T = \textbf{track}(\textbf{a}$). There is a constant $\delta$ depending on $\mathbf{S}$, $\text{dom}(D)$ and $w$, such that if $\vert P \vert>\delta$, then $P$ contains at least two T-isomorphic facts.
\begin{proof}
In each tree of $\mathcal{W}(D,\Sigma)$, facts can be constructed by permuting at most $w + \text{dom}(D)$ terms (at most $w$ new labelled nulls and $\text{dom}(D)$ possible constants) over $w$ positions of $\vert \mathbf{S} \vert$ facts. For each tree we have at most $\vert \mathbf{S} \vert (w + \text{dom}(D))^w$ non $T$-isomorphic facts. Tracks are ${\vert \mathbf{S} \vert (w + \text{dom}(D))}\choose{2}$, therefore $\delta$ = ${\vert \mathbf{S} \vert (w + \text{dom}(D))}$$ {\vert \mathbf{S} \vert (w + \text{dom}(D))}\choose{2}$.
\end{proof}
\end{theorem}

We are now ready to discuss an implementation of $\textit{chase}^H$ that uses our relaxed warded semantics. Hence, by redefining  $\textit{chase}^H(D,\Sigma)$ as $\textit{chase}(\textit{chase}^{B_W}(D,\Sigma_T),\Sigma_E)$, we prove the following result.

\begin{theorem}
\label{th:chaseBW-CQ} 
Given a set $\Sigma = \Sigma_T \cup \Sigma_E$ of warded TGDs and harmless EGDs, and a database $D$, if $D\cup\Sigma$ is satisfiable, for every query $Q$, it holds $\textit{chase}(D,\Sigma)\models Q$ iff
$\textit{chase}(\textit{chase}^{B_W}(D,\Sigma_T),\Sigma_E)\models Q$. 
\end{theorem}
\begin{proof}
It is sufficient to prove that (a)~$\textit{chase}^B(D,\Sigma_T) \models Q$ iff $\textit{chase}^{B_W}(D,\Sigma_T) \models Q$ and then apply the same proof as Theorem~\ref{th:decidability}, where $\textit{chase}^{B_W}(D,\Sigma)$ is used in the place of $\textit{chase}^{B}(D,\Sigma)$.
We prove the two directions of the implication.

\smallskip\noindent
($\Rightarrow$) We need to show that for every homomorphism $h$ from $Q$ to $\textit{chase}^B(D,\Sigma)$, there is a corresponding homomorphism from $Q$ to $\textit{chase}^{B_W}(D,\Sigma)$. Consider a conjunction of facts $\boldsymbol{\varphi}(\mathbf{x})$ in $\textit{chase}^B(D,\Sigma)$.
If $\boldsymbol{\varphi}(\mathbf{x})$ does not contain any pair of $T$-isomorphic facts, then $\boldsymbol{\varphi}(\mathbf{x}) \in \textit{chase}^{B_W}(D,\Sigma)$. The only case in which $\boldsymbol{\varphi}(\mathbf{x}) \not\in \textit{chase}^{B_W}(D,\Sigma)$ is when $\boldsymbol{\varphi}(\mathbf{x})$ contains two T-isomorphic facts $\varphi_i,\varphi_j$ and so only one of them, e.g., $\varphi_i$ is in $\textit{chase}^{B_W}(D,\Sigma)$. In this case, we argue that for every homomorphism from $Q$ to $\boldsymbol{\varphi}(\mathbf{x})$, there is a homomorphism from $Q$ to $\boldsymbol{\varphi}^\prime(\mathbf{x}) = \boldsymbol{\varphi}(\mathbf{x}) \{\varphi_j / \varphi_i\}$
($\boldsymbol{\varphi}^\prime(\mathbf{x})$ is obtained from $\boldsymbol{\varphi}(\mathbf{x})$ by replacing $\varphi_j$ with $\varphi_i$). Let us proceed by contradiction and assume that there does not exist a homomorphism from $Q$ to $\boldsymbol{\varphi}^\prime(\mathbf{x})$. Since $\varphi_i$ and $\varphi_j$ are $T$-isomorphic, they differ at most by some labelled nulls $z_i \neq z_j$, in corresponding positions. In order for $z_i$ and $z_j$ to prevent the existence of homomorphisms from $Q$ to $\boldsymbol{\varphi}^\prime(\mathbf{x})$, $z_i$ and $z_j$ must appear also in other facts of $\boldsymbol{\varphi}(\mathbf{x})$ and $\boldsymbol{\varphi}^\prime(\mathbf{x})$. Let us suppose there is a fact $\varphi_k$ of $\boldsymbol{\varphi}(\mathbf{x})$ where $z_i$ appears in position $\varphi_k[i]$ and assume there is a corresponding fact $\varphi_k^\prime$ of $\boldsymbol{\varphi}^\prime(\mathbf{x})$ s.t.\ $z_j$ is not in position $\phi_k^\prime[i]$. Since $\varphi_i$ and $\varphi_k$, $\varphi_j$ and $\varphi_k^\prime$ share a labelled null, respectively, by wardedness it must be that either $\varphi_k$ resp.\ $\varphi_k^\prime$ is a predecessor of $\varphi_i$ resp.\ $\varphi_j$ in the chase, or it is a successor. In the former case, as $\varphi_i$ and $\varphi_j$ are $T$-isomorphic ($\textbf{track}(\varphi_i)=\textbf{track}(\varphi_j)$), by Theorem~\ref{th:t-isomoprhism_implies_subtree_isomorphism}, $\varphi_k$ and $\varphi_k^\prime$ must be T-isomorphic, which contradicts our hypothesis that $z_i$ appears in position $\phi_k[i]$ and $z_j$ is not in position $\phi_k^\prime[i]$. In the latter case, from $T$-isomorphism of $\varphi_i$ and $\varphi_j$, we have that $\textit{subgraph}(\mathcal{G},\varphi_i)$ is isomorphic to $\textit{subgraph}(\mathcal{G},\varphi_j^\prime)$, and thus $\varphi_k$ and $\varphi_k^\prime$ are $T$-isomorphic ($\mathcal{G}$ is the common chase graph), which also contradicts our hypothesis. Thus we conclude that if there exists a homomorphism from $Q$ to $\boldsymbol{\varphi}(\mathbf{x})$ then there exists a homomorphism from $Q$ to $\boldsymbol{\varphi}^\prime(\mathbf{x})$.

\noindent($\Leftarrow$) It directly descends from the hypothesis that $\textit{chase}^{B_W}(D,\Sigma)\subseteq  \textit{chase}^B(D,\Sigma) $ and the chase monotonicity.
\end{proof}

\subsection{Reasoning Algorithms in the Presence of EGDs}
\label{sec:algorithms}

Decidability guarantees given by Theorem~\ref{th:decidability} are already enough to safely apply the standard EGD chase~\cite{FKMP05}, where for each chase step enforcing a TGDs, all EGDs are applied to fixpoint. For every fact that is generated by a TGD, a set of unifications potentially arise and all of them are applied. This is helpful when the interaction between TGDs and EGDs is not known beforehand and EGDs are applied to enable other TGDs to fire.
On the other hand, harmless EGDs and warded TGDs allow for a smarter and approach. First a full $\textit{chase}^{B_W}(D,\Sigma)$ is computed, with the aid of a space-efficient support structure to assess $T$-isomorphism and thus implement an aggressive recursion control to generate the chase; then, EGDs can be applied to fixpoint. Finally, CQs can be answered on the obtained result.
Towards an implementation of harmless EGDs in the {\sc Vadalog} system, we propose an algorithm that practically implements $\textit{chase}^H(D,\Sigma)$. 

\smallskip\noindent\textbf{Algorithm.} The main idea of the algorithm consists in first computing $\textit{chase}^{B_W}(B,\Sigma)$ and then expanding it by applying the EGDs of $\Sigma_E$ to fixpoint so as to enable CQ answering. 
To construct $\textit{chase}^{B_W}(B,\Sigma)$, we adopt a variant of the usual restricted chase procedure, where applicability is conditioned to the success of a \textit{termination strategy}, a function which ensures that for a class of equivalent facts modulo $T$-isomorphism, we generate only one copy.

\begin{definition}
\label{def:termination_strategy}
\textit{Given a database $D$, a set of warded rules $\Sigma$, and a binary relation $R$ between facts, a \emph{termination strategy} $\mathcal{T}_R$ is a function $\mathcal{T}_R: \textit{chase}^k(D,\Sigma) \to \{0,1\}$, where $k$ denotes the chase generated until the $k$-th chase step, such that $\mathcal{T}_R(f) = {1}$ if there is no fact $f^\prime \in \textit{chase}^{k-1}(D,\Sigma)$ with $fRf^\prime$, 0 otherwise.}
\end{definition}

With Definition~\ref{def:termination_strategy} in place, we can constructively define $\textit{chase}^{B_W}(D,\Sigma)$ as follows. Starting with $\textit{chase}^0 = D$ at a given chase step $k$, a TGD $\boldsymbol{\phi}_{k-1}(\mathbf{x}) \to \exists\mathbf{z}~\boldsymbol{\psi}(\mathbf{y},\mathbf{z})$
is applied if there is a homomorphism $\theta$ such that $\theta(\boldsymbol{\phi}) \subseteq \textit{chase}^{k-1}(D,\Sigma)$ and $\mathcal{T}_R(\theta(\boldsymbol{\phi}))=1$, where $R$ is $T$-isomorphism of facts.

\begin{enumerate}
    \item build $\textit{chase}^{B_W}(D,\Sigma_T)$ as $\mathcal{M}$, by applying the chase restricted by $\mathcal{T}_R$, where $R$ is $T$-isomorphism of facts
    \item initialize $\mathcal{G} = \langle \mathbf{N},\mathbf{E} \rangle$ as an empty graph 
    \item  until the following joint conditions hold: (a) there is a an EGD $\eta = \boldsymbol{\phi}(\mathbf{x}) \to x_i = x_j$ of $\Sigma_E$ such that there exists a homomorphism $\theta$ with $\theta(\boldsymbol{\phi}) \subseteq \mathcal{M}$ and $\theta(x_i)\neq\theta(x_j)$; (b) $\textit{memory\_size}(\mathcal{G})<{\rm threshold}$
        \begin{itemize}
            \item let $\mathbf{P} = \{\langle \theta(x_i), \theta(x_j) \rangle\}$ be the set of pairs s.t.\ $\theta(\boldsymbol{\phi}) \subseteq \mathcal{M}$ and $\theta(x_i)\neq\theta(x_j)$
            \item for each pair $\langle \mathbf{a},\mathbf{b} \rangle \in \mathbf{P}$
            \begin{itemize}
            \item if both $\mathbf{a}$ and $\mathbf{b}$ are constants then \textbf{fail}
            \item if $\mathbf{a}$ (resp.\ $\mathbf{b}$) is a labelled null and until $\textit{memory\_size}(\mathcal{G})<{\rm threshold}$
            \begin{itemize}
            \item[-] add nodes $\mathbf{a}$ and $\mathbf{b}$ to $\mathbf{N}$ (unless already included) and add edge $\langle \mathbf{a},\mathbf{b} \rangle$ to $\mathbf{E}$
            \end{itemize}
            \end{itemize}
        \end{itemize}
    \item for each $CC \in {\textit{connected\_components}}({\mathcal{G}})$, if there are two constant nodes $\mathbf{a}\neq\mathbf{b}\in\textit{CC}$, then fail; if there is a constant node $\mathbf{a} \in CC$, then assign the value of $\mathbf{a}$ to all the labelled null nodes; else, assign all the nodes to an arbitrary labelled null corresponding to a node in $CC$.
    \item apply all the assignments in $\mathcal{G}$ to $\mathcal{M}$.
\end{enumerate}

\noindent\textbf{Discussion}. Arguments about correctness of the algorithm are quite intuitive: step~(1) builds $\textit{chase}^{B_W}(D,\Sigma)$ by enforcing the $T$-isomorphism equivalence relation; in step~(3), EGDs are applied to fixpoint and bindings are simulated by edges of $\mathcal{G}$. Since by harmlessness, no EGD can reactivate any TGD, $\textit{chase}^{B_W}(D,\Sigma)$ is finite by definition as well as the number of possible unifications applied by (4) in $\mathcal{G}$. We are therefore guaranteed that the algorithm terminates. Throughout the construction of $\mathcal{G}$, labelled nulls are progressively unified until all the ones that are directly or indirectly connected by the EGDs take the same value. The addition of new edges in the graph introduces new bindings but never alters the existing ones, unless any contradiction arises resulting in a hard EGD violation.
The bottleneck of the algorithm lies in computing the initial chase, which dominates the EGD unification.
Because EGD unification can be executed in constant time on polynomially many processors, depending on the specific subfragment of Datalog$^{\pm}$ in which TGDs of $\Sigma_T$ are specified, the algorithm may enjoy high parallelizability. For instance, while Warded Datalog$^{\pm}$ is not highly paralellizable, it is the case for an interesting and quite expressive variant of it, namely \emph{Piecewise Linear Warded Datalog$^{\pm}$}~\cite{BGPS19}, for which CQ answering is in $\textsf{NL}$ in data complexity.

\smallskip
The algorithm guarantees limited memory footprint. Line~(3) needs to evaluate the body of an EGD $\eta$ against $\textit{chase}^{B_W}(D,\Sigma)$, which has been pre-computed at previous steps. This can be done in constant space, by allocating a minimum of one memory block for each of the involved relations.  Since the query of step~(3) is repeatedly executed, in order to progress in the algorithm execution, we need to unify at least one $\langle \mathbf{a}, \mathbf{b} \rangle$ pair at a time, which can be stored in constant space (in one memory block). From an implementation point of view, there is a trade-off between the number of evaluations of the body of each EGD and the allocated memory for $\mathcal{G}$. In fact, multiple activations of the same EGD against the same database can be typically factored out in such a way that physical block reads are minimized, along the lines of usual query answering in database systems, where buffering is largely adopted. On the other hand, small available memory for $\mathcal{G}$ causes either more EGD evaluations or increased footprint to store intermediate results.

\section{Conclusion}
\label{sec:conclusion}
The aim of this work is pushing the boundaries of EGDs, fostering their adoption in a much broader context. We introduced a new fragment, namely harmless EGDs, that enables to fully use the unification power of equality dependencies, so far neglected in the existing fragments. In conjunction with decidable and efficient fragments for TGDs, like Warded Datalog$^\pm$, harmlessness shows not to hamper either decidability or complexity of query answering, while increasing the fragment expressive power. 
In a practical perspective, we studied an intuitive syntactic condition, which is sufficient to witness harmlessness. We paved the way for a practical implementation of the fragment by proposing a new chase variant, with a restriction condition based on the joint theoretical underpinnings of harmless EGDs and warded TGDs.

\smallskip Harmless EGDs open extensive future work. Subfragments of harmless EGDs turn out to have been already largely adopted, unknowingly, in a host of existing industrial scenarios as well as chase and data exchange benchmarks. Beyond this, we believe that harmless EGDs can be at the basis of a new generation of knowledge graph traversal algorithms, where problems requiring the expressive power of transitive closure are solved in a scalable way thanks to the unification power of this new fragment. 
Besides a full-fledged implementation in the {\sc Vadalog} systems, additional and interesting theoretical problems are still to be pursued. It is our conjecture, for example, that harmless EGDs do not hamper complexity of CQ answering, independently of the TGD fragment. We will report about these evolutions soon.

\bibliographystyle{plain}
\bibliography{b}

\begin{thebibliography}{10}

\bibitem{AbHV95}
Serge Abiteboul, Richard Hull, and Victor Vianu.
\newblock {\em Foundations of Databases}.
\newblock Addison-Wesley, 1995.

\bibitem{ArGP14}
Marcelo Arenas, Georg Gottlob, and Andreas Pieris.
\newblock Expressive languages for querying the semantic web.
\newblock In {\em PODS}, pages 14--26, 2014.

\bibitem{ArenasGP18}
Marcelo Arenas, Georg Gottlob, and Andreas Pieris.
\newblock Expressive languages for querying the semantic web.
\newblock {\em {ACM} Trans. Database Syst.}, 43(3):13:1--13:45, 2018.

\bibitem{AGCM15}
Patricia~C. Arocena, Boris Glavic, Radu Ciucanu, and Ren{\'{e}}e~J. Miller.
\newblock The ibench integration metadata generator.
\newblock {\em {PVLDB}}, 9(3):108--119, 2015.

\bibitem{BFGK18}
Luigi Bellomarini, Ruslan~R. Fayzrakhmanov, Georg Gottlob, Andrey Kravchenko,
  Eleonora Laurenza, Yavor Nenov, St{\'{e}}phane Reissfelder, Emanuel
  Sallinger, Evgeny Sherkhonov, and Lianlong Wu.
\newblock Data science with vadalog: Bridging machine learning and reasoning.
\newblock In {\em {MEDI}}, volume 11163 of {\em Lecture Notes in Computer
  Science}, pages 3--21. Springer, 2018.

\bibitem{BeSG18}
Luigi Bellomarini, Emanuel Sallinger, and Georg Gottlob.
\newblock The vadalog system: Datalog-based reasoning for knowledge graphs.
\newblock {\em {PVLDB}}, 11(9):975--987, 2018.

\bibitem{BKMM17}
Michael Benedikt, George Konstantinidis, Giansalvatore Mecca, Boris Motik,
  Paolo Papotti, Donatello Santoro, and Efthymia Tsamoura.
\newblock Benchmarking the chase.
\newblock In {\em PODS}, pages 37--52, 2017.

\bibitem{BGPS19}
Gerald Berger, Georg Gottlob, Andreas Pieris, and Emanuel Sallinger.
\newblock The space-efficient core of vadalog.
\newblock In {\em {PODS}}, pages 270--284. {ACM}, 2019.

\bibitem{CaCF13}
Andrea Cal{\`{\i}}, Marco Console, and Riccardo Frosini.
\newblock Deep separability of ontological constraints.
\newblock {\em CoRR}, abs/1312.5914, 2013.

\bibitem{CaGK08}
Andrea Cal{\`{\i}}, Georg Gottlob, and Michael Kifer.
\newblock Taming the infinite chase: Query answering under expressive
  relational constraints.
\newblock In {\em Description Logics}, volume 353 of {\em {CEUR} Workshop
  Proceedings}. CEUR-WS.org, 2008.

\bibitem{CaGK13}
Andrea Cal\`{\i}, Georg Gottlob, and Michael Kifer.
\newblock Taming the infinite chase: Query answering under expressive
  relational constraints.
\newblock {\em J. Artif. Intell. Res.}, 48:115--174, 2013.

\bibitem{CaGL09}
Andrea Cal{\`{\i}}, Georg Gottlob, and Thomas Lukasiewicz.
\newblock A general datalog-based framework for tractable query answering over
  ontologies.
\newblock In {\em PODS}, pages 77--86, 2009.

\bibitem{CaGL12}
Andrea Cal\`{\i}, Georg Gottlob, and Thomas Lukasiewicz.
\newblock A general {D}atalog-based framework for tractable query answering
  over ontologies.
\newblock {\em J. Web Sem.}, 14:57--83, 2012.

\bibitem{CGOP12}
Andrea Cal{\`{\i}}, Georg Gottlob, Giorgio Orsi, and Andreas Pieris.
\newblock On the interaction of existential rules and equality constraints in
  ontology querying.
\newblock In {\em Correct Reasoning}, volume 7265 of {\em Lecture Notes in
  Computer Science}, pages 117--133. Springer, 2012.

\bibitem{CaGP10}
Andrea Cal{\`{\i}}, Georg Gottlob, and Andreas Pieris.
\newblock Advanced processing for ontological queries.
\newblock {\em Proc. {VLDB} Endow.}, 3(1):554--565, 2010.

\bibitem{CaGP12}
Andrea Cal\`{\i}, Georg Gottlob, and Andreas Pieris.
\newblock Towards more expressive ontology languages: The query answering
  problem.
\newblock {\em Artificial Intelligence}, 193:87--128, 2012.

\bibitem{CaPi11}
Andrea Cal{\`{\i}} and Andreas Pieris.
\newblock On equality-generating dependencies in ontology querying -
  preliminary report.
\newblock In {\em {AMW}}, volume 749 of {\em {CEUR} Workshop Proceedings}.
  CEUR-WS.org, 2011.

\bibitem{CaVa81}
Beeri Catriel and Moshe~Y. Vardi.
\newblock The implication problem for data dependencies.
\newblock {\em Automata, Languages and Programming}, pages 73--85, 1981.

\bibitem{datalog2}
Stefano Ceri, Georg Gottlob, and Letizia Tanca.
\newblock What you always wanted to know about datalog (and never dared to
  ask).
\newblock {\em TKDE}, 1(1):146--166, 1989.

\bibitem{ChVa85}
Ashok~K. Chandra and Moshe~Y. Vardi.
\newblock The implication problem for functional and inclusion dependencies is
  undecidable.
\newblock {\em {SIAM} J. Comput.}, 14(3):671--677, 1985.

\bibitem{dantsin}
Evgeny Dantsin, Thomas Eiter, Georg Gottlob, and Andrei Voronkov.
\newblock Complexity and expressive power of logic programming.
\newblock {\em {ACM} Comput. Surv.}, 33(3):374--425, 2001.

\bibitem{FKMP05}
Ronald Fagin, Phokion~G. Kolaitis, Ren{\'e}e~J. Miller, and Lucian Popa.
\newblock Data exchange: {S}emantics and query answering.
\newblock {\em Theor. Comput. Sci.}, 336(1):89--124, 2005.

\bibitem{FKPT05}
Ronald Fagin, Phokion~G. Kolaitis, Lucian Popa, and Wang~Chiew Tan.
\newblock Composing schema mappings: Second-order dependencies to the rescue.
\newblock {\em {ACM} Trans. Database Syst.}, 30(4):994--1055, 2005.

\bibitem{GoPi15}
Georg Gottlob and Andreas Pieris.
\newblock Beyond {SPARQL} under {OWL} 2 {QL} entailment regime: Rules to the
  rescue.
\newblock In {\em IJCAI}, pages 2999--3007, 2015.

\bibitem{MaMS79}
David Maier, Alberto~O. Mendelzon, and Yehoshua Sagiv.
\newblock Testing implications of data dependencies.
\newblock {\em ACM Transactions on Database Systems}, 4(4):455--468, 1979.

\end{thebibliography}
\end{document}